\documentclass[10pt,a4paper]{article}%
\usepackage{amssymb,amsmath, amsfonts,bm}
\usepackage{graphicx,graphics}
\usepackage{mathtools}
\usepackage{bbold}
\usepackage[english]{babel}
\usepackage[utf8]{inputenc}
\usepackage{epsfig,url} 
\usepackage{bbm,theorem}
\usepackage{a4wide}
\usepackage{color}
\usepackage{enumerate}
\usepackage{calrsfs}
\usepackage[bookmarks=true]{hyperref}
\usepackage{bookmark}
\usepackage{xcolor}

\usepackage[affil-it]{authblk}
\usepackage{amssymb}
\usepackage{verbatim} 
\usepackage{graphicx}%
\providecommand{\U}[1]{\protect\rule{.1in}{.1in}}
\DeclareMathAlphabet{\pazocal}{OMS}{zplm}{m}{n}
\newtheorem{theorem}{Theorem}[section]
\newtheorem{definition}[theorem]{Definition}

\newtheorem{lemma}[theorem]{Lemma}
\newtheorem{proposition}[theorem]{Proposition}
\newtheorem{assumption}[theorem]{Assumption}

\newtheorem{notation}[theorem]{Notation}

{\theorembodyfont{\upshape}
\newtheorem{remark}[theorem]{Remark}

\numberwithin{equation}{section}
\numberwithin{theorem}{section}

\providecommand{\Cyl}{\operatorname{Cyl}}
\newcommand{\qed}{\hfill$\Box$}

\newenvironment{proof}{\begin{trivlist}\item[]{\em Proof:}\/}{\qed\end{trivlist}}

\renewcommand{\geq}{\geqslant}
\renewcommand{\leq}{\leqslant}

\newcommand{\ep}{\varepsilon}

\newcommand{\R}{\mathbb{R}}
\newcommand{\N}{\mathbb{N}}

\newcommand{\PP}{\mathbb{P}}
\newcommand{\E}{\mathbb{E}}

\newcommand{\T}{\mathcal{T}}
\newcommand{\LL}{\mathcal{L}}

\newcommand{\ud}{\;\mathrm{d}}
\newcommand{\uud}{\mathrm{d}}
\newcommand{\mbf}{\bm}

\newcommand{\Reals}{{\mathbb R}}

\newcommand{\M}{\mathcal{M}}

\providecommand{\X}{\mathcal{X}}

\providecommand{\eeR}{{\eta_{\varepsilon}}_{|R} }

\providecommand{\eeRn}{{\eta_{\varepsilon, n}}_{|R}}

\newcommand{\eps}{\varepsilon}

\newcommand{\de}{\mathop{}\!\mathrm{d}}
\newcommand{\A}{\mathbb{A}}

\usepackage{mathrsfs}

\usepackage{latexsym}

\begin{document}

\title{Kinetic description of a Rayleigh Gas with annihilation}

\date{\today}

\author[1]{Alessia Nota \thanks{\texttt{nota@iam.uni-bonn.de }}}
\author[2]{Raphael Winter \thanks{\texttt{raphaelwinter@iam.uni-bonn.de}}} 
\author[3]{Bertrand Lods \thanks{\texttt{bertrand.lods@unito.it }}}

\affil[1,2]{\em Institute for Applied Mathematics, University of Bonn, \em Endenicher Allee 60, D-53115 Bonn, Germany}  
\affil[3] {\em Dipartimento ESOMAS, Universit\`a degli Studi di Torino \& Collegio Carlo Alberto,\em Corso Unione Sovietica, 218/bis, 10134 Torino, Italy}

\maketitle

\begin{abstract}
In this paper, we consider the dynamics of a tagged point particle in a gas of moving hard-spheres that are non-interacting among each other. This model is known as the ideal Rayleigh gas.     We add to this model the possibility of annihilation (ideal Rayleigh gas with annihilation), requiring that each obstacle is either annihilating or elastic, which determines whether the tagged particle is elastically reflected or removed from the system. We provide a rigorous derivation of a linear Boltzmann equation with annihilation from this particle model in the Boltzmann-Grad limit. Moreover, we give explicit estimates for the error in the kinetic limit by estimating the contributions of the configurations which prevent the Markovianity. 
The estimates show that the system can be approximated by the Boltzmann equation on an algebraically long time scale in the scaling parameter.
\end{abstract}

\tableofcontents

\section{Introduction}

We consider an extension of the Boltzmann-Rayleigh gas particle model that includes the annihilation of particles. From this model we derive the Boltzmann-Rayleigh equation in the Boltzmann-Grad limit, and give explicit error estimates. For annihilation rate $\alpha\in [0,1]$, the equation for the probability density of the limit process reads
\begin{align} \label{eq:linBoltzannihil}
\partial_t f + v \cdot\nabla_x f + \mu \alpha \lambda(v) f = \mu (1-\alpha) Q(\mathcal{M}_{\beta},f) .
\end{align}  
Here $\mu>0$ is a fixed constant  proportional to the inverse of the mean free path and $\beta>0$ is the inverse temperature of the background. The linear Boltzmann operator $Q(\mathcal{M}_{\beta},f)$ and the collision frequency $\lambda$ are defined as 
\begin{align}
Q(\mathcal{M}_{\beta},f)(v)&= \int_{\R^3} \int_{\mathbb{S}^{2} } [(v-v_1)\cdot \hat n]_+  \left( \mathcal{M}_{\beta}(v_1^{\prime}) f(v^{\prime})-\mathcal{M}_{\beta}(v_1)f(v) \right)\ud{v_1} \ud{\hat n}	, \\
\lambda(v)&=\lambda_0 \int_{\mathbb R^3}dv_1 \mathcal{M}_{\beta}(v_1)|v-v_1|, \quad \lambda_0=\int_{\mathbb{S}^{2}}\!\de \hat n \, [\hat n\cdot\hat v]_+, \quad \hat v\in \mathbb{S}^{2} \label{eq:lambdadef},
\end{align}		
with $v,v'$, $v_1,v_1'$ given by the elastic hard-sphere collision rule. Notice that $\lambda_0$ is independent of $\hat v$ due to the rotational invariance of the integral. 

The Boltzmann equation with annihilation \eqref{eq:linBoltzannihil} comes from the corresponding nonlinear Boltzmann equation with ballistic annihilation introduced in \cite{BKL, CDT}. This equation and related models have been introduced to extend the theory of statistical mechanics to systems without conservation of particles, e.g. systems with chemical reactions.  The first rigorous mathematical results on the nonlinear Boltzmann equation with positive annihilation probability $\alpha>0$ prove existence and uniqueness of self-similar solutions as well as their stability (cf. \cite{ABL,jde,jde2}). Recently, a first rigorous derivation of the spatially homogeneous nonlinear Boltzmann equation with annihilation from a Kac-like particle system has been obtained in \cite{LNP}, under suitable assumption on the collision kernel.

In this paper, we provide a rigorous derivation of the Boltzmann-Rayleigh equation \eqref{eq:linBoltzannihil} in an inhomogeneous setting for the case of hard-sphere interactions. 
We start from a particle model that consists of a single tagged point particle interacting with a heat bath of hard-sphere obstacles of radius $\eps>0$ in $\Reals^3$. The tagged particle and the obstacles are assumed to have equal mass, and the initial distribution of obstacles is a grand-canonical distribution with a rescaled  average of  $\mu_\ep=\mu\eps^{-2}$ particles per unit volume.
Every obstacle is initially randomly chosen to be either annihilating with probability $\alpha$, or elastic with probability $1-\alpha$.
The obstacles do not interact among each other, but undergo hard-sphere collisions with the tagged particle if they are elastic. If the tagged particle collides with an annihilating obstacle, it passes to the annihilated state $\A$, and is thus removed from the system.

It is important to remark that we recover the classical Boltzmann-Rayleigh equation in the case $\alpha=0$. This equation has been obtained from Rayleigh gas models in \cite{BLLS} and \cite{S2}, and recently under more general assumptions in \cite{MST}. In the present paper, we present a derivation (cf. Section \ref{s:conv}) of equation \eqref{eq:linBoltzannihil}, and we further provide explicit estimates for the set of pathological configurations, that prevent the process from being Markovian for $\eps>0$ (cf. Section \ref{s:rec}). The estimate is valid for all values $\alpha \in [0,1]$, and is the most delicate part of the analysis. The bound for these pathologies is algebraically vanishing in $\eps \rightarrow 0$, and thus opens the possibility of deriving the diffusion equation which gives the hydrodynamic description on a longer time scale. In Section~\ref{s:hydro} we propose some possible long-time asymptotics in the case $\alpha>0$. 

We observe that the technique we use to provide the qualitative validation of \eqref{eq:linBoltzannihil} (cf. Section~\ref{s:conv}) builds on the constructive approach originally proposed by Gallavotti to obtain the linear Boltzmann equation from a Lorentz gas of hard-spheres (cf. \cite{G}). This approach has later been extended to more general interaction potentials and different physical situations (see for instance \cite{BNP, BNPP, DP, MN, N, NSV}). 
We emphasize that in contrast to the Lorentz gas, the energy of a tagged particle in the Boltzmann-Rayleigh gas is not constant. In particular, the expected time to the next collision, i.e. the inverse of the collision frequency $\lambda(v(t))$, changes during the evolution. This prevents the explicit resumming of the collision expansion as used in the Gallavotti argument. We circumvent this issue by comparing the error of the approximation to the moments of the number of collisions of the limit process (cf. Section \ref{ss:2mom}), similar to the arguments in \cite{BGS} and \cite{NV}.

The Boltzmann-Rayleigh equation and its hydrodynamic limit have also been derived from a slightly different model in \cite{BGS} on the unit torus. In this model, the background obstacles interact among each other, which allows perturbations due to the tagged particle to perpetuate in the heat bath. Due to the rapidly growing collision tree, the analysis is closer to the derivation of the nonlinear Boltzmann equation and yields logarithmic error estimates in the scaling parameter $\eps\rightarrow 0$.

\section{The model and main results}

We consider the dynamics of a tagged point particle in a gas of hard-sphere obstacles. We assume that the obstacles have radius $\eps>0$ and form a gas at thermal equilibrium distributed in the whole space $\Reals^3$.
The dynamics of the system is given by interactions of the tagged particle with the moving obstacles, the obstacles do not interact among each other. This model is referred to as the ideal Rayleigh gas (cf. \cite{S2}). In the following we extend the model by the possibility of annihilation. Each obstacle is either annihilating or elastic, which determines whether the tagged particle is elastically reflected or removed from the system upon collision with it.

The velocity space is then $\R^3$, and the reference measure is the Maxwellian distribution with temperature $\beta^{-1}>0$, whose density with respect to the Lebesgue measure is denoted by $\mathcal{M}_{\beta}(v)$:
$$\mathcal{M}_{\beta}(v)=\left(\frac{\beta}{2\pi}\right)^{\frac{3}{2}}\exp\left(-\frac{\beta}{2}|v|^{2}\right), \qquad v \in \R^{3}.$$ 
We denote by $\mbf{c}=(c_1,\dots,c_q, \ldots)$ the centers $c_{i} \in \R^{3}$ of the countable set of scatters and by $\mbf{w}=(w_1,\dots,w_q, \dots)$ their corresponding velocities. To include annihilating obstacles, we further consider a sequence $\mbf{z}=(z_1,\ldots,z_i,\ldots)$, $z_i\in \{0,1\}$. Here $z_i=0$ denotes that the obstacle $(c_i,w_i,z_i)$ is annihilating, similarly it is elastic if $z_i=1$.

\medskip

\begin{notation} \label{not:bernoulli}
	For $\alpha\in [0,1]$ we denote by $b_\alpha$ the Bernoulli measure on $\{0,1\}$ which is given by $b_\alpha(\{0\})=\alpha$, $b_\alpha(\{1\})=(1-\alpha)$. 
\end{notation}	

\begin{definition} \label{def:annihilated}
Denote by $\X$ the three dimensional phase space with an annihilated state $\A$, i.e. $\X = (\Reals^3 \times \Reals^3) \cup \{\A\}$. 
	We extend the metric of the phase space $\R^3\times \R^3$ by $|\A-(x,v)|=\infty$ for all $(x,v)\in\R^3\times \R^3$. We equip $\X$ with the measure $\ud{m}=\mathbf{1}_{\R^{3}\times\R^{3}}\ud{x}\ud{v}+\delta_{\A}$ where $\ud{x}\ud{v}$ denotes the Lebesgue measure on $\Reals^3\times \Reals^3$ and $\delta_{\A}$ the Dirac measure in $\A$.
\end{definition}

We will now specify the probability distribution yielding the initial datum of the system.

The obstacles will be given by a Point Process $\eta$ on $X:=\Reals^3\times \Reals^3 \times \{0,1\}$, and the tagged particle be distributed according to some smooth probability density $f_0$ on the phase space $\X$.

On the function $f_0(x,v)$ 
defining the initial probability distribution for the tagged particle,  we assume the following: 
\begin{assumption}\label{ass:f0} 
	Let $f_0:\X\to \R^{+}$ be such  that 
	$f_0\in L^1_{x}(\R^3; L^{\infty}_{\mathcal{M}_{\beta}^{-1}}(\R^3))$ and $f_0(\A)=0$.
	More precisely:
	\begin{align*}
	\int_{\R^{3}} \|f_0(x,\cdot) \mathcal{M}^{-1}_\beta(\cdot)\|_{L^{\infty}(\Reals^3)} \ud{x} <\infty.
	\end{align*}
\end{assumption}

In the sequel, random variables are defined with respect to a given fixed probability space $\left({\Omega},\mathcal{F},\PP\right)$.  
\begin{definition}[Boltzmann-Rayleigh gas] \label{def:BRgas}
Let $X:=\Reals^3\times \Reals^3 \times \{0,1\}$. 
	Fix $\eps, \mu,\beta>0$, and $\alpha\in [0,1]$. Let $f_0\in L^1(\X)$ be as in Assumption~\ref{ass:f0}. We define the intensity measure 
	$$\nu(\uud{c}\uud{w}\uud{z})=\mu\uud{c}\,\M_{\beta}(w)\uud{w}\,b_\alpha(\uud{z}).$$
	Consider a \emph{random variable} $(x_0,v_0)$ on $\Reals^3 \times \Reals^3$ and a
	\emph{point process} $\eta=(c_j,w_j,z_j)_{j\in J}$ on $X$, where $J$ is some countable index set. Let
	the joint distribution of $\eta$ and $(x_0,v_0)$ be given by:
	\begin{align} \label{eq:jointDistr}
		\PP(\eta(A_1)=\ell_1,\ldots \eta(A_n)=\ell_n, (x_0,v_0) \in D)= \int_{D} f_0(x,v) \prod_{k=1}^n\frac{\nu(A_k^\eps)^{\ell_k} }{\ell_k!} e^{-\nu(A_k^\eps)}\ud{x} \ud{v}.
	\end{align}
	Here $D\subset \Reals^3 \times \Reals^3$ is a Borel set, $A_k \subset \Reals^3 \times \Reals^3 \times \{0,1\}$, $1\leq k \leq n$ are mutually disjoint Borel sets that are bounded in the first variable.  Moreover, $\eta(A_k)$ denotes the number of obstacles in $A_k$, and for fixed $x\in \Reals^3$ we use the notation
	$A_k^\eps:= A_k\setminus (B_\eps(x)\times \Reals^3 \times \{0,1\})$.
\end{definition}
In the above definition and in all the sequel, unless otherwise specified, we identify a point process $\eta$ on $X$ with its realization $\eta=(c_j,w_j,z_j)_{j\in J}$, with $J$  countable.

We assume that each obstacle of the configuration has radius $\eps>0$.  We now define the evolution of the system.
\begin{definition}[Dynamics of the Rayleigh gas with annihilation]\label{def:particleflow}
	We define the evolution of the particle system as follows. Let $(x,v)\in \X$ the initial position and velocity of the tagged particle, and $\eta=(c_j,w_j,z_j)_{j\in J}$ be  a given initial obstacle configuration. Then $(x(t),v(t))\in \X$ for $t>0$ is determined by the following deterministic dynamics:
	\begin{enumerate}
	\item if $(x(t^{\ast}),v(t^{\ast}))=\A $ for some $0\leq t^{\ast}< t$ then $(x(t),v(t))=\A$.
	\item  if $(x(t^{\ast}),v(t^{\ast}))\in\R^3\times \R^3 $ for all $0\leq t^{\ast}< t$  and $|x(t_-)-c_j(t_-)|>\eps$ for all $j \in J$, then we solve the following system of ODEs:
	\begin{align}
		&\dot{x}(t)=v, \;\;\dot{v}(t)=0, \quad   \quad \text{(tagged particle)}\label{eq:pflow}\\& 
		\dot{c}_j(t)=w_j,\;\; \dot{w}_j (t)=0, \quad   \quad \text{(background)}\label{eq:obflow}
	\end{align} 
 subject to the boundary conditions:
	\begin{enumerate}[(i)]
		\item If $|x(t_-)-c_j(t_-)|= \eps$, $z_j=1$ for some $j\in J$: 
			\begin{align}
				v(t)=v(t_-)- \frac{1}{\eps^2} \big(v(t_-)-w_j(t_-)\big)\cdot \big(x(t_-)-c_j(t_-)\big) (x(t_-)-c_j(t_-)) , \\
				w_j(t)=v(t_-)- \frac{1}{\eps^2}\big(v(t_-)-w_j(t_-)\big)\cdot \big(x(t_-)-c_j(t_-)\big) (x(t_-)-c_j(t_-)) ,	
			\end{align}
		\item If $|x(t_-)-c_j(t_-)|= \eps$, $z_j=0$ for some $j\in J$:
			\begin{align}
				(x(t),v(t))=\A.
			\end{align}	
	\end{enumerate}
	\end{enumerate}
	For the induced mapping of the phase space of the tagged particle we write:
	\begin{align}
		T_{\eta}^{t}(x,v) = (x(t),v(t)).
	\end{align}
	Here, as above, $\eta=(c_j,w_j,z_j)_{j\in J}$ denotes the initial obstacle configuration.
\end{definition}
\begin{remark}
	The definition above is well-posed almost surely. Indeed, it is possible to prove that, with probability one, the tagged particle only collides with one obstacle at any given time, and only finitely many obstacles in finite time.  This has been shown in the literature (see for instance \cite{A}). 
	We remark that the well-posedness, as proved in \cite{A}, is not affected by the introduction of the annihilated state $\A$ since the tagged particle stays in the annihilated state for all times after arriving there, and no further interactions can occur.  

\end{remark}

As we stated above, we assume that each obstacle of the configuration has radius $\eps>0$. We then rescale the intensity $\mu$ of the centers of the obstacles as follows.

\begin{definition}[Boltzmann-Grad Scaling]
	We consider the following scaling limit of the intensity $\mu_{\eps}$:
	\begin{align}
		\mu_{\varepsilon}&=\varepsilon^{-2}\mu. \label{scaling:B}	
	\end{align} 
We denote by $\nu _{\ep}$  the rescaled intensity measure $\nu$, with $\mu$ replaced by $\mu_\ep$. $\E_{\ep}$ will be the expectation with respect to the  Poisson point process $\eta_{\ep}$ with intensity measure $\nu_{\eps}$. 
\end{definition}
\bigskip
 

 We are interested in the evolution of the one-particle correlation function
$f_{\eps}(t,x,v)$ of the tagged particle induced by the dynamics given in Definition \ref{def:particleflow}. 
More precisely, let $T^t_{\eta_{\eps}}(x,v)$ be the phase space location of the tagged particle at time $t$ that starts at $(x_0,v_0)=(x,v)$ for a given obstacle configuration $\eta_{\eps}$.  
Then we define $f_{\ep}$ in the weak formulation as 
\begin{equation}\begin{aligned}\label{def:fep}
\int_{\X}\phi f_{\ep}(t) \ud{m}&= \E[\phi(T^t_{\eta_{\eps}}(x_0,v_0))], \quad \phi \in C_{b}(\X). 
\end{aligned}\end{equation}
In what follows we will use the notation 
$$\displaystyle \langle \phi,f_{\ep}(t)\rangle= \int_{\X} \phi f_\eps(t) \ud{m}.$$ 
We emphasize that the argument in \cite{A} shows that the resulting  function $f_{\eps}\in C([0,T],L^1(\X))$ for any $T>0$. 

\medskip

The main result of the present paper can be summarized in the following theorem.
\begin{theorem}\label{th:MAIN1}
Let 
$f_0$  satisfy Assumption \ref{ass:f0} and let $f_{\ep}$ be defined as in \eqref{def:fep}.
For all $T>0$ we have
\begin{equation*}
\lim_{\ep\to 0}\left\|f_{\ep}-f\right\|_{C([0,T]; L^{1}(\X))}=0
\end{equation*}
where $f$ is the solution to the linear Boltzmann equation with annihilation 
\begin{equation}\label{eq:linearized}
\left\{\begin{array}{ll}\vspace{2mm}
	\partial_t f + v \cdot \nabla_x f + \mu \alpha \lambda(v) f = \mu (1-\alpha) Q(\mathcal{M}_{\beta},f) &\\
	f(x,v,0)=f_0(x,v).&
\end{array}\right.
\end{equation}	
Here the linear Boltzmann operator is defined as 
\begin{equation}\label{eq:opQ}
Q(\mathcal{M}_{\beta},f)(v)= \int_{\R^3} \int_{\mathbb{S}^{2} } [(v-v_1)\cdot \hat n]_+  \left( \mathcal{M}_{\beta}(v_1^{\prime}) f(v^{\prime})-\mathcal{M}_{\beta}(v_1)f(v) \right)\ud{v_1} \ud{\hat n}		
\end{equation}					 
and $(v,v_1)$ is  a  pair  of  incoming velocities   and  $(v^{\prime},v_1^{\prime})$ is  the
corresponding pair of outgoing velocities defined by the elastic reflection with transferred momentum in direction $\hat n$:
\begin{equation}\label{scattering}
v'= v-\hat n\cdot(v-v_1)\,\hat n,\qquad v_1'= v_1+\hat n\cdot(v-v_1)\,\hat n.
\end{equation} 

Moreover, let be $\kappa_0>0$, $r\in (0,\frac 1 3)$ and $\kappa\in (0,\kappa_0)$. Then there exists a constant $C>0$ such that
\begin{align*}
	\|f_{\eps}(t)-f(t)\|_{L^1(\X)} \leq C \left[{\ep}^{\kappa} t + \eps^{r-\kappa_0} ( t+t^2)+\eps^{r-2\kappa_0}( t^\frac 5 2+ t^\frac72)]\right], \qquad \forall t \geq0.
\end{align*}
\end{theorem}

\begin{remark}
	We remark that the theorem above yields the classical Rayleigh-Boltzmann equation in the case of non-annihilating particles, i.e. $\alpha=0$.
\end{remark}
\begin{remark}
	The function $\lambda$ (cf. \eqref{eq:lambdadef}) is bounded away from $0$ and it has linear growth for 
	large  $|v|$.
	Therefore  $|v|$ and $|v|^2/\lambda(v)$ have all the exponential moments with respect to $\M_{\beta}(v)\de v$. For further details we refer to 
	\cite{BBB}.
\end{remark}

From now on, we write for simplicity $\LL f=Q(\M_{\beta},f)$ and, as well-known, $\LL$ can be split into a gain and loss term
$$\LL f=\LL^{+}f-\LL^{-}f$$
with 
$$\LL^{-}f(v)=\lambda(v)f(v)$$
where $\lambda(v)$  is the scattering rate given by \eqref{eq:lambdadef}. Moreover, the gain part $\LL^{+}$ can be written as an integral operator thanks to the \emph{Carleman representation} (see for instance \cite{AL,carleman}) resulting in 
\begin{equation}\label{eq:kern}
\LL^{+}f(v)=\int_{\R^{3}}k(v',v)f(v')\ud{v'}, \qquad \text{ with } \quad k(v',v)=\frac{1}{|v-v'|}\int_{E(v',v)}\mathcal{M}_{\beta}(u)\ud{u}.
\end{equation}
Here $E(v,v')$ is the hyperplane orthogonal to $v'-v$ passing through $v'$ (and~$\ud{u}$ is the Lebesgue measure over that hyperplane):
\begin{align}\label{eq:E}
E(v,v'):=\{u \in \Reals^3: (v'-v)\cdot (v'-u)=0\}. 
\end{align} 

\section{From the particle system to the Boltzmann equation with annihilation: convergence of the one-particle correlation function}\label{s:conv}

In this section we present our strategy to prove Theorem \ref{th:MAIN1}. In Section \ref{ss:series}, we first provide a semi-explicit series form of the solution $f(t,x,v)$ to the Boltzmann equation with annihilation  \eqref{eq:linearized}. This allows us to use a direct approach to obtain the Markovian approximation. Indeed, in Section \ref{ss:strategy} we will compare the microscopic solution $f_{\eps}(t,x,v)$ given by \eqref{def:fep} to the series solution $f(t,x,v)$ of \eqref{eq:linearized} and identify the error term which we will estimate in Section~\ref{s:rec}.

\subsection{Series solution of the Boltzmann equation with annihilation }\label{ss:series}

We assume that the initial probability distribution of the  test particle $f_0$ satisfies Assumption \ref{ass:f0}. Consider the linear Boltzmann equation with annihilation rate $\alpha>0$:
\begin{equation}\label{eq:Bolg}
(\partial_{t}+v\cdot \nabla_{x})f(t,x,v)=\mu \left[(1-\alpha)\LL- \alpha \LL^{-} \right]f(t,x,v)=\mu \left[(1-\alpha)\LL^{+}-\LL^{-} \right]f(t,x,v). 
\end{equation} 
Introduce the $C_{0}$-semigroup $(S(t))_{t\geq0}$ with generator given by the transport and absorption operators $(v\cdot\nabla_x+\mu \lambda(v)I)$,
$$S(t)f(x,v):=f(x-vt,v)e^{-\mu \lambda(v)t}, \qquad (x,v) \in \R^{3}\times\R^{3}, \quad t\geq0.$$
We infer the following formula for a solution to equation \eqref{eq:linearized}:
\begin{equation*}
f(t,x,v)=S(t)f_0(x,v)+\mu (1-\alpha)\int_{0}^{t}S(t-s)(\LL^+ f)(x,v,s)\ud{s}.
\end{equation*} 
By iteration we find the following series expansion 
\begin{align} \label{eq:Duhamel1}
f(t)= S(t)f_0+\sum_{n=1}^\infty (1-\alpha)^n \mu ^n \int_{0}^{t}\ud{t_1}\ldots \int_{0}^{t_{n-1}}\ud{t_n}
\, S(t-t_1) \LL^+ \ldots \LL^+ S(t_{n})f_0.
\end{align}
Now, using the Carleman representation \eqref{eq:kern}, we get 
\begin{multline}\label{eq:srsolB_1}
f(t,x,v)= \sum_{n=0}^\infty (1-\alpha)^n \mu^n \int_{0}^{t}\ud{t_1} \int_{\Reals^3}k(v_{1},v)\ud{v_{1}}\ldots\int_{0}^{t_{n-1}}\ud{t_{n}}\int_{\Reals^{3}}k(v_{n},v_{n-1})\ud{v_{n}}
\\
\ldots f_{0}(x-\sum_{i=0}^{n-1} v_i(t_i-t_{i+1})-t_n v_n,v_n)e^{-\mu\sum_{i=0}^{n-1}\lambda(v_i)(t_{i}-t_{i+1})} e^{-\mu\lambda(v_n) t_n}.
\end{multline}
We observe that 
in \eqref{eq:srsolB_1} we used the following ordering for collision times: 
$$0\leq t_n\leq \dots \leq t_1\leq t=: t_0$$ and 
we denoted as $v_{i}$ the outgoing velocity (along the backward trajectory), i.e. the velocity of the tagged particle after suffering a collision at time $t_i$. 
The following Lemma will be useful throughout the analysis. 
\begin{lemma}\label{lem:e}
Let $E(v,v')$ be given by \eqref{eq:E}. 
	For all $v^{\prime}\neq v\in \Reals^3$ we have:
	\begin{align} \label{eq:IntEst}
	\int_{E(v,v')} \mathcal{M}_{\beta}(u) \ud{u} = \mathcal{M}_{\beta}\left(v'\cdot \frac{v'-v}{|v'-v|}\right) =:e(v,v') . 
	\end{align}
\end{lemma}
\begin{proof}
	By definition, the set $E(v,v')$ is given by the equation
	\begin{align*}
	\frac{v'-v}{|v'-v|} \cdot (v'-u) = 0.
	\end{align*}
	Hence $E(v,v')$ is a hyperplane with distance $v'\cdot \frac{v'-v}{|v'-v|}$ from the origin.
	Chosing an orthonormal basis including the normal vector to $E(v,v')$, explicit integration yields the desired identity.
\end{proof}
With this notation and using \eqref{eq:srsolB_1}  we get 
\begin{equation}
\begin{aligned}\label{eq:srsolB_3}
&f(t,x,v)=  \sum_{n=0}^\infty (1-\alpha)^n \mu^n \int_{0}^{t}\ud{t_1}\dots \int_{0}^{t_{n-1}}\ud{t_n}\int_{\Reals^3 } \ud{v_1}\ldots \int_{\Reals^3 } \ud{v_n} \\ \ldots
&  f_0(x-\sum_{i=0}^{n-1} v_i(t_{i}-t_{i+1})-t_n v_n,v_n)\left(\prod_{i=1}^n \frac{e(v_i,v_{i-1})}{|v_i-v_{i-1}|}  \right)e^{-\mu\sum_{i=0}^{n-1}\lambda(v_i)(t_{i}-t_{i+1})}e^{-\mu\lambda(v_n)t_n}.
\end{aligned} 
\end{equation}
We will use the formulas \eqref{eq:srsolB_1} and \eqref{eq:srsolB_3} to compare
the distribution $f_{\ep}$ of the tagged particle, in the scaling limit, to the solution $f$ of the limiting equation.

\subsection{Strategy}\label{ss:strategy}

We first give an equivalent definition of the function $f_{\eps}$. In order to do this, we observe that the space $ C_{b}(\X)$ is the set of functions of the form
\begin{equation*}
\phi(x,v)=\tilde{\phi}(x,v)\quad \text{on}\quad \R^3\times \R^3 \quad \text{and} \quad \phi(\A)= \sigma \in \R ,
\end{equation*} 
with $\tilde{\phi}\in C_{b}(\R^3\times\R^3)$. Therefore, it is straightforward to prove the following general result.
\begin{proposition}\label{prop:deffeps}
Let $g$ be a function in $L^1(\X)$. Then $g$ is uniquely determined by
\begin{equation}\label{def:g}
\langle \phi,g\rangle=\int_{\X} \phi \cdot g \ud{m}\quad \text{for}\quad \phi\in C_b(\X), \;\phi(\A)=0,\quad \text{and} \quad g(\A)=\sigma\in \Reals.
\end{equation}
\end{proposition}
Notice that Proposition \ref{prop:deffeps} implies that $f_{\eps}$, defined as in \eqref{def:fep}, can be equivalently defined by
\begin{align}\label{def:fep1}
&\langle \phi,f_{\ep}(t)\rangle
			=\int_{\X} \E\left[ \phi(T^{t}_{\eta_{\eps}}(x,v))|(x_0,v_0)=(x,v)\right]f_0(x,v) \ud{m}, \quad \phi\in C_b(\X), \;\phi(\A)=0\\&
			f_{\eps}(t,A) = 1-\|f_{\eps}(t)\|_{L^{1}(\R^3\times \R^3)}.\label{def:fepA}
\end{align}

In what follows we will identify an obstacle configuration $( c_{q}, w_{q},z_q)_{q\in J}$ with its
empirical measure. Further, for a given initial position $(x,v)$ of the tagged particle, we introduce the spatial localization $\eeR$ to the ball $B_R(x)$:
\begin{align}
	\eta_{\eps}(\uud{y}\uud{v}\uud{z}) &= \sum_{j\in J} \delta(y-c_j) \delta (v-w_j) \delta(z-\alpha_j) \\
	\eeR (\uud{y}\uud{v}\uud{z})			&= \sum_{j\in J: c_j\in B_R(x)} \delta(y-c_j) \delta (v-w_j) \delta(z-\alpha_j).
\end{align}
The function   $f_{\ep}(x,v,t)$ given as in \eqref{def:fep} then satisfies, for all $\phi \in C_b(\X) $: 
\begin{align*}
\langle \phi, f_\eps(t)\rangle=	\int_{\X} \phi f_\eps (t)\ud{m} = \lim_{R\rightarrow \infty} \int_{\X} \E[(\phi \circ T^t_{\eeR})|(x_0,v_0)=(x,v)  ] f_0(x,v)\ud{m},
	\end{align*}
	with the measure $m$ introduced in Definition \ref{def:annihilated}. 
	We observe that the identity above holds as a consequence of the Borel-Cantelli Lemma applied to the family $\{(\eta_{\ep};x_0,v_0)\,:\, x(s)\notin B_R(x),\, s\in [0,t] \}$.
Due to Proposition \ref{prop:deffeps} in what follows it is sufficient to consider as test functions the functions $\phi \in C_b(\X) $ such that $\phi(\A)=0$.
Courtesy of the localization, and $f_0(\A)=0$, we can write the expectation using \eqref{eq:jointDistr}:
 \begin{align*}
\langle \phi, f_\eps(t)\rangle=&\lim_{R\to\infty}\int_{\R^3\times \R^3} \ud{x} \ud{v}\, e^{-\mu_{\ep}|B^{\ep}_R(x)|}\nonumber \\& \quad \quad \sum_{q\geq 0}\frac{\mu_{\ep}^{q}}{q!}\int_{(B^{\ep}_R(x))^q}\uud{\mbf{c}}_{q}\int_{(\R^3)^q} \uud{\mbf{w}}_q \mathcal{M}_{\beta}(\mbf{w}_q)\int_{\{0,1\}^q}\uud{\mbf{z}}_qb_{\alpha}(\mbf{z}_q) (\phi \circ T^t_{\eeR}) f_0 .
\end{align*} 
Here we use the notation $B^{\ep}_R:=B_R(x)\setminus B_{\ep}(x)$ where $B_R(x)$ and $B_{\ep}(x)$ denote the balls centered in $x$ with radius $R$ and $\ep$ respectively.  This ensures the condition that the tagged particle does not start from an obstacle.

We distinguish the obstacles of the configuration $({c}_{q},{w}_{q},z_q)_{q}$ which, up to the time $t$, influence the motion, called internal obstacles, and the external ones. More precisely,
an obstacle  $(c_i,w_i,z_i)$ is internal if
\begin{equation*}
\inf_{0\leq s\leq t}|x(s)-(c_i+w_i s)|\leq\ep,
\end{equation*}
and is called external otherwise, i.e. if
\begin{equation*}
\inf_{0\leq s\leq t}|x(s)-(c_i+w_i s)|>\ep.
\end{equation*} 

For simplicity we write  $\mbf{c}_{q}=(c_1,\dots, c_q)$, $\mbf{w}_{q}=(w_1,\dots ,w_q)$ and $\mbf{z}_q=(z_1,\dots, z_q)$. Then, we decompose a given configuration $(\mbf{c}_{q},\mbf{w}_{q}, \mbf{z}_q)=(\mbf{b}_{n},\mbf{u}_{n},\mbf{z}_n)\cup (\mbf{\tilde{b}}_p,\mbf{\tilde{w}}_p,\tilde{\mbf{z}}_p )$ where $(\mbf{b}_{n},\mbf{u}_{n},\mbf{z}_n)$ is the set of all the \textit{internal} obstacles and $(\mbf{\tilde{b}}_p,\mbf{\tilde{w}}_p,\tilde{\mbf{z}}_p )$ the set of all the \textit{external} ones. From now on $\mathbf{1}_{\{\cdot\}}$ will denote the characteristic function of the set or event $\{\cdot\}$. 
Since the $q$ obstacles can be grouped into $n$ internal and $p=q-n$ external obstacles in
${q \choose n}$ different ways, we can use this decomposition to write:
\begin{align*}
\langle \phi, f_\eps(t)\rangle&=  \lim_{R\to\infty}\int_{\R^3\times\R^3 }\ud{x} \ud{v} \, e^{-\mu_{\ep}|B^{\ep}_R|}\\&
\quad \sum_{n\geq 0}\frac{\mu_{\ep}^{n}}{n!}\int_{(B^{\ep}_R)^n}\uud{\mbf{b}}_{n}\int_{(\R^3)^n} \uud{\mbf{u}}_{n} \mathcal{M}_{\beta}(\mbf{u}_n)\int_{\{0,1\}^n}d\mbf{z}_n\,b_{\alpha}(\mbf{z}_n) \mathbf{1}_{\{{\eeRn} \; \text{internal}\}}\\&
\quad  \sum_{p\geq 0}\frac{\mu_{\ep}^{p}}{p!}\int_{(B^{\ep}_R(x))^p}\uud{\mbf{\tilde{b}}}_{p}\int_{(\R^3)^p} d\mbf{\tilde{w}}_p \mathcal{M}_{\beta}(\mbf{\tilde{w}}_p)\int_{\{0,1\}^p}\uud{\mbf{z}}_p\,b_{\alpha}(\mbf{z}_p)  \mathbf{1}_{\{ {\eta_{\varepsilon, p}}_{|R} \; \text{external}\}}\phi \circ T^t_{{\eeR}}\,f_0 .
\end{align*}
We notice that $T^{t}_{{\eeR}}(x,v)=T^{t}_{{\eeRn}}(x,v)$. Further notice that $T^{t}_{\eeRn}(x,v)=\A$ if there is an annihilating internal obstacle. Moreover, since we require $\phi\in C_b(\X)$ with $\phi(\A)=0$ 
we can also perform the integrals in  $z_i=1$ for all internal obstacles. Then, integrating over the external obstacles and taking the limit $R\to \infty$, yields:
\begin{align}\label{eq:fep_intermediate}
 \langle \phi,f_{\ep}(t)\rangle=& \sum_{n\geq 0}\int_{\R^3\times\R^3}\ud{x} \ud{v} \, (1-\alpha)^{n}\frac{\mu_{\ep}^{n}}{n!}\int_{(\Reals^3_\eps)^n}\uud{\mbf{b}}_{n}\,\int_{(\R^3)^n} \uud{\mbf{u}}_n \mathcal{M}_{\beta}(\mbf{u}_n)\, \mathbf{1}_{(\{ ( \mbf{b}_n,\mbf{u}_n )_{n}\;\text{internal}\})} \nonumber \\&
\quad
\exp\left({-\mu_{\ep} \int_{\R^3} du \mathcal{M}_{\beta}(u) \, |\T_{t}(x,v,u;{\eta}_{\ep,n})|}\right)  \, 
 f_0(x,v) \phi(T^{t}_{{\eta_{\ep,n}}}(x,v)) ,
\end{align}
where $\T_{t}(x,v,w;{\eta}_{\ep,n})$ is the tube defined as follows (see Figure \ref{fig1}). 
\begin{figure}[th]
\label{fig1}\centering
\includegraphics [scale=0.3]{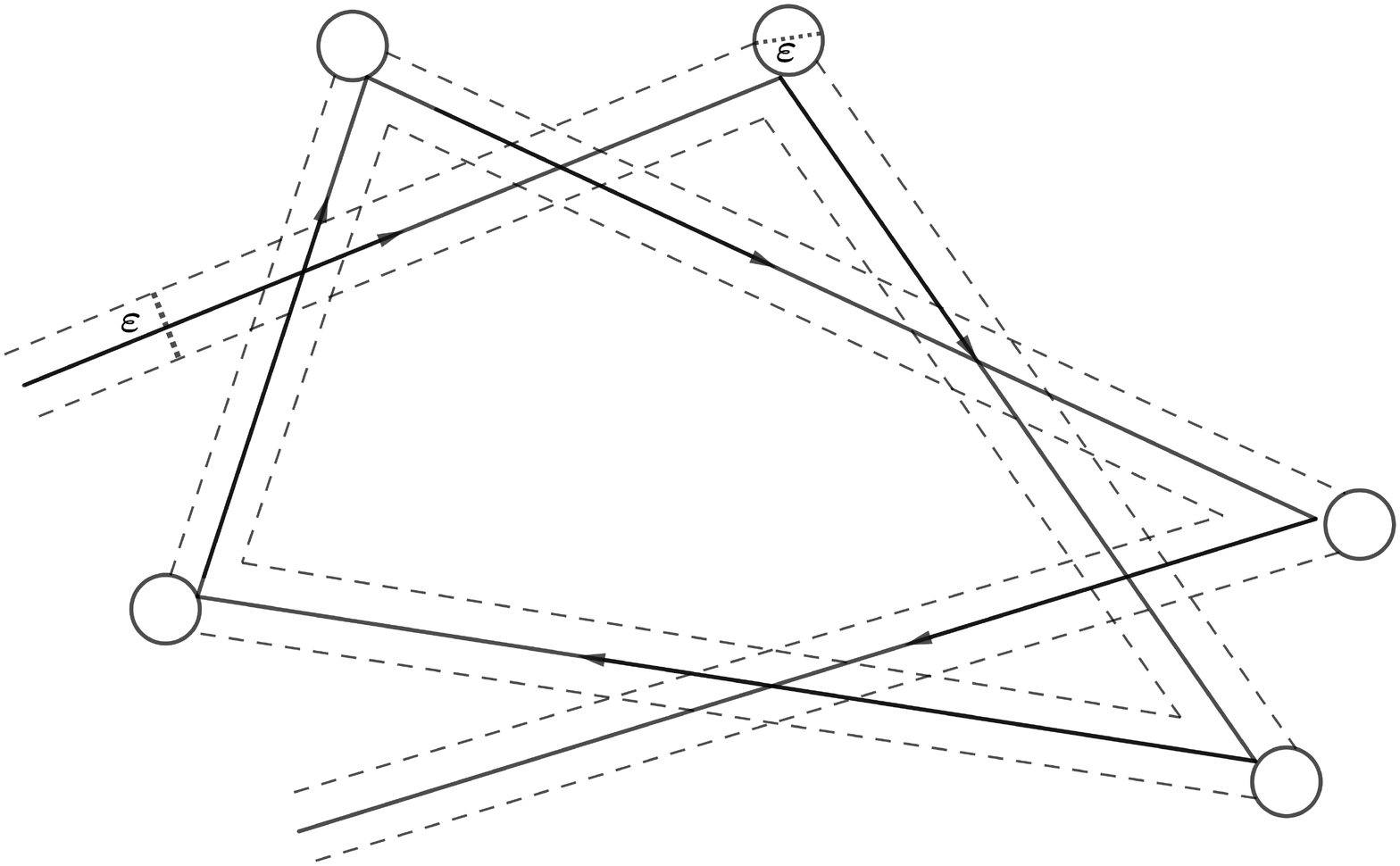}\caption{The dynamical tube $\T_{t}(x,v,w; {\eta}_{\ep,n})$ spanned by the tagged particle} %
\end{figure}

	For a positive time $t>0$,  $x,v\in \Reals^3$, a tuple of obstacles $\eta_{\ep,n}$ and a velocity $w\in \Reals^3$ we define
	\begin{align} \label{eq:colltube}
		\mathcal{T}_{t}(x,v,w;\eta_{\ep,n})= \{y\in \Reals^3: \exists s\in [0,t] \text{ s.t. } |y+w s-x(s)|\leq \eps\}.
	\end{align}

\medskip	
	In order to prove Theorem \ref{th:MAIN1}, namely the convergence of $f_{\eps}$ towards the solution $f$ of the Rayleigh Boltzmann equation with annihilation \eqref{eq:linearized} we need the following auxiliary Definitions and Propositions.
	
Note that, according to a classical argument introduced in \cite{G} (see also  \cite{BNP}, \cite{BNPP}, \cite{DP}, \cite{MN}, \cite{N}), we remove from ${f}_{\ep}$ all the bad events, namely those untypical with respect to the Markov process described by $f$. Then we will show they have low probability. Following this idea, 
 we build now a function $\tilde{f}_{\eps}$ from \eqref{eq:fep_intermediate} that forgets all trajectories that involve multiple collisions with the same obstacle:
\begin{definition}\label{def:tildefep}
Let $f_0$ satisfy Assumption \ref{ass:f0}. For any $\eps> 0$ and any $t \geq 0$, we will say that a given configuration of \emph{internal} obstacles $(\mbf{b}_n,\mbf{u}_n)_{n}$ belongs to $\Delta_{n}(t)$ if each obstacle is hit exactly \emph{once} by the trajectory $s \mapsto T^{s}_{{\eta_{\ep,n}}}(x,v)$ in the time interval $[0,t]$. We define the function $\tilde{f}_{\eps}(t)\in L^1(\X)$ as follows
\begin{align}\label{eq:tildefep1}
 \langle \phi,\tilde{f}_{\ep}(t)\rangle=& \sum_{n\geq 0}(1-\alpha)^{n}\frac{\mu_{\ep}^{n}}{n!}\int_{\R^3\times\R^3}\ud{x} \ud{v} \,\int_{(\Reals^3_\eps)^n}\uud{\mbf{b}}_{n}\,\int_{(\R^3)^n} \uud{\mbf{u}}_n \mathcal{M}_{\beta}(\mbf{u}_n)\, \mathbf{1}_{\{(\mbf{b}_n,\mbf{u}_n )_{n} \in \bm{\Delta}_{n}(t)\}}\nonumber \\&
\quad
\exp\left({-\mu_{\ep} \int_{\R^3} du \mathcal{M}_{\beta}(u) \, |\T_{t}(x,v,u;{\eta}_{\ep,n})|}\right)  \, 
 f_0(x,v) \phi(T^{t}_{{\eta_{\ep,n}}}(x,v)) ,
\end{align}
for all $\phi \in C_b(\X) $ with $\phi(\A)=0$ and set
\begin{align}\label{eq:tildefep1A}
\tilde{f}_{\ep}(t,\A)=& \sum_{n\geq 0}\left(1-(1-\alpha)^{n}\right)\int_{\R^3\times\R^3}\ud{x} \ud{v} \, \int_{(\Reals^3_\eps)^n}\uud{\mbf{b}}_{n}\,\int_{(\R^3)^n} \uud{\mbf{u}}_n \mathcal{M}_{\beta}(\mbf{u}_n)\, \mathbf{1}_{\{(\mbf{b}_n,\mbf{u}_n )_{n} \in \bm{\Delta}_{n}(t)\}}\nonumber \\&
\quad
\exp\left({-\mu_{\ep} \int_{\R^3} du \mathcal{M}_{\beta}(u) \, |\T_{t}(x,v,u;{\eta}_{\ep,n})|}\right)  \, 
 f_0(x,v) .
\end{align}
\end{definition}
	
\begin{remark}
We notice that $\bm{\Delta}_{n}(t)=\bm{\Delta}_{n}(t;x,v)$ depends on the initial position and velocity of the tagged particle and it satisfies 
$\mathbf{1}_{\{(\mbf{b}_n,\mbf{u}_n )_{n} \in \bm{\Delta}_{n}(t)\}}=\mathbf{1}_{(\{ ( \mbf{b}_n,\mbf{u}_n )_{n}\;\text{internal}\})}\mathbf{1}_{\{\rm{precisely}\; n \;\rm{ collisions}\}}$.
\end{remark}


\begin{definition}\label{def:limitflow}
Let $(x,v)\in \R^3\times \R^3$ and let $\{t_i\}_{i=0}^n$ be a sequence of positive collision times such that $0\leq t_{n}\leq t_{n-1}\leq \dots \leq t_2\leq t_1\leq t_0=t$ and $v=v_0\to v_1\to \dots \to v_n$ a sequence of velocities. We define the backward limit flow 
as follows: 
for $s\in [t_{i+1},t_{i})$ we set 
\begin{equation}
\gamma^{-s}(x,v)=\left(x-\sum_{k=0}^{i-1} v_k(t_{k}-t_{k+1})-(t_i-(t-s))v_i,v_i\right),
\end{equation}
and, in particular, for $s=t$ we have
\begin{equation}
\gamma^{-t}(x,v)=\left(x-\sum_{k=0}^{n-1} v_k(t_{k}-t_{k+1})-t_n v_n,v_n\right).
\end{equation}
Moreover, we will write, for any $s\in [0,t]$, $\gamma^{-s}(x,v)=(x(-s),v(-s))$ to denote the first and second component respectively. 
\end{definition}

\begin{definition}\label{def:pathologies}
We define the set of pathological configurations as follows: 
\begin{itemize}
\item[i)] \textbf{Recollisions:}\\ There exists $b_i(-s)$ such that for some $s\in(t-t_{j},t-t_{j+1})$, $j>i$, $x(-s)\in \partial B(b_i(-s),\ep)$.
\item[ii)] \textbf{Interferences:}\\ There exists $b_j(-s)$ such that $x(-s)\in B(b_j(-s),\ep)$ for some $s\in(t-t_{i},t-t_{i+1})$, $j>i$.
\end{itemize}
Here we used the ordering of collision times $0\leq t_{n}\leq t_{n-1}\leq \dots \leq t_2\leq t_1\leq t_0=t.$ 
Further we denoted by $x(-s)$ the first component of the backward limit flow given as in Definition \eqref{def:limitflow} and by $b_i(-s)$ the position of the obstacles along the backward in time evolution.
\end{definition}

\begin{proposition}\label{prop:step1strategy}
Let $f_0$ satisfy Assumption \ref{ass:f0} and let  $ f_{\eps}$  be defined as in \eqref{def:fep} (or, equivalently, as in \eqref{def:fep1}-\eqref{def:fepA}).   
The function $\tilde{f}_{\eps}$ defined in \eqref{eq:tildefep1}-\eqref{eq:tildefep1A} satisfies 
$$0 \leq \bar{f}_{\ep}(t) \leq \tilde{f}_{\ep}(t) \leq f_{\ep}(t)$$
where 
$\bar{f}_{\ep}\in L^1(\X)$ is defined as 
\begin{equation}
\begin{aligned} \label{formula6}
\bar{f}_{\ep}(t,x,v)&=\sum_{n=0}^\infty  (1-\alpha)^{n}\mu^n \int_{\Omega'_n}\uud{t}_1 \uud{v}_1\uud{\xi'_{1}}\ldots \uud{t}_n \uud{v}_n\uud{\xi'_{n}} \left(\prod_{i=1}^n \frac{\mathcal{M}_{\beta}(\xi'_i)}{|v_i-v_{i-1}|}\right)  \\
&\quad e^{- \mu \sum_{i=0}^{n-1} \lambda(v_i)(t_{i}-t_{i-1})}\,e^{- \mu\lambda(v_n)t_n}  \,(1-\bm{\mathbf{1}}_{\rm{rec}})(1-\bm{\mathbf{1}}_{\rm{int}}) f_0(\gamma^{-t}(x,v))
\end{aligned}
\end{equation}
and  
\begin{align}\label{formula6A}
	\bar{f}_{\ep}(t,\A)&=\int_{\R^3\times \R^3} \ud{x} \ud{v}\, \sum_{n=0}^\infty  \left(1-(1-\alpha)^{n}\right)\mu^n \int_{\Omega'_n}\uud{t}_1 \uud{v}_1\uud{\xi'_{1}}\ldots \uud{t}_n \uud{v}_n\uud{\xi'_{n}}  \left( \prod_{i=1}^n \frac{\M_{\beta}(\xi'_i)}{|v_i-v_{i-1}|}\right)\nonumber \\&
	\quad e^{- \mu \sum_{i=0}^{n-1} \lambda(v_i)(t_{i}-t_{i+1})}\,e^{- \mu\lambda(v_n)t_n}  \,(1-\bm{\mathbf{1}}_{\rm{rec}})(1-\bm{\mathbf{1}}_{\rm{int}}) f_0(\gamma^{-t}(x,v)) .
\end{align}	
Here we used the notation
 $$ \int_{0}^{t_0}\uud{t}_1\int \uud{v}_1\int_{E(v_1,v)} \uud{\xi'_1}  \dots  \int_{0}^{t_{n-1}}\uud{t}_n \int \uud{v}_n \int_{E(v_n,v_{n-1})} \uud{\xi'_n}= \int_{\Omega'_n}\uud{t}_1 \uud{v}_1\uud{\xi'_{1}}\ldots \uud{t}_n \uud{v}_n\uud{\xi'_{n}} ,$$
 with $E(v_i,v_{i-1})$ as in  \eqref{eq:E}. Moreover, $\bm{\mathbf{1}}_{\rm{rec}}$ and ${\bm{\mathbf{1}}}_{\rm{int}}$ denote the indicator functions of the sets of pathological configurations i) and ii) introduced in Definition \ref{def:pathologies} and $\gamma^{-t}(x,v)$ is the backward limit flow defined in Definition \ref{def:limitflow}. 
\end{proposition}

\begin{proof} 
 	For an obstacle configuration $\eta_{\ep,n}$ and an initial phase-space position $(x,v)$, assume that 
 	the tagged particle experiences collisions at times $\tau_i$, $i=1,\ldots,n$ between
 	the times $\tau_0=0$ and 
	$t$ and has velocity $v_i$ in the time interval $[\tau_i,\tau_{i+1}]$. Notice that here we are looking at the forward evolution, hence the different notation used to denote collision times ($\tau_i$ instead of $t_i$) along the forward flow. 
	
	Then we can estimate the size of the collision tube defined as in \eqref{eq:colltube} by:
	\begin{align}\label{eq:boundTube}
		|\mathcal{T}_t (x,v,u;\eta_{\ep,n})| \leq \sum_{i=0}^{n} \eps^2 \pi |v_i-u| (\tau_{i+1}-\tau_i).
	\end{align}
	We now consider the argument of the exponential in \eqref{eq:tildefep1} 
	and, using \eqref{eq:boundTube} as well as the definition of the function $\lambda(v)$ introduced in \eqref{eq:lambdadef}, we obtain:
	\begin{align}\label{eq:tubeintupper}
		\mu_\eps \int \M_{\beta}(u) |\mathcal{T}_t (x,v,u;\eta_{\ep,n})|\ud{u} \leq \mu_\eps \eps^2\sum_{i=0}^n\lambda(v_i) (\tau_{i+1}-\tau_i)=\mu \sum_{i=0}^n\lambda(v_i) (\tau_{i+1}-\tau_i).
	\end{align}
We now introduce, for any $t \geq0$, $h_\eps$
as follows:
\begin{align}\label{eq:intermediate}
\langle \phi,h_\eps(t)\rangle:=& \int_{\R^3\times\R^3 } \ud{x} \ud{v} \, \sum_{n\geq 0}(1-\alpha)^{n}\frac{\mu_{\ep}^{n}}{n!} \int_{(\Reals^3_\eps)^n}\uud{\mbf{b}}_{n}\,\int_{(\R^3)^n} \uud{\mbf{u}}_{n} \mathcal{M}_{\beta}(\mbf{u}_n)\,\mathbf{1}_{\{(\mbf{b}_n,\mbf{u}_n)_{n}\in \bm{\Delta}_{n}(t)\}}\nonumber \\&
\exp\left({-\mu \sum_{i=0}^n\lambda(v_i) (\tau_{i+1}-\tau_i)}\right)  \,  f_0(x,v) \phi(T^{t}_{{\eta_{\ep,n}}}(x,v)),
\end{align}
for all $\phi \in C_b(\X) $ with $\phi(A)=0$, and  
\begin{align}\label{eq:barfep1A}
	h_{\ep}(t,A)&=\int_{\R^3\times \R^3} \ud{x} \ud{v}\, \sum_{n=0}^\infty  \left(1-(1-\alpha)^{n}\right)\mu^n \frac{\mu_{\ep}^{n}}{n!} \int_{(\Reals^3_\eps)^n}d\mbf{b}_{n}\,\int_{(\R^3)^n} d\mbf{u}_n \mathcal{M}_{\beta}(\mbf{u}_n)\, \nonumber \\&
	\qquad  \mathbf{1}_{\{(\mbf{b}_n,\mbf{u}_n)_{n}\in \bm{\Delta}_{n}(t)\}} 
	\exp\left({-\mu \sum_{i=0}^n\lambda(v_i) (\tau_{i+1}-\tau_i)}\right)  \,  f_0(x,v) .
\end{align}	
Thanks to the estimate above (cf.~\eqref{eq:tubeintupper}), we can handle the exponential factor in \eqref{eq:tildefep1} and we get 
 $0\leq h_\eps \leq \tilde{f}_{\ep} \leq f_\eps$. It remains to prove that $h_{\ep}=\bar{f}_{\ep}$ given as in \eqref{formula6}-\eqref{formula6A}.
 
Since, thanks to $\mathbf{1}_{\{(\mbf{b}_n,\mbf{u}_n)_{n}\in \bm{\Delta}_{n}(t)\}}$, we restrict to the set of realizations for which the tagged particle collides with
each of the obstacles $(b_i,u_i)$ precisely once (cf. Definition \ref{def:tildefep}), we can order the obstacles
according to the scattering sequence, i.e. $(b_i,u_i)$ is collided before the obstacle $(b_j,u_j)$ if $i<j$. Then we perform the following change of variables  \begin{equation}
\label{change var}
(b_1,u_1),\dots,(b_n,u_n) \rightarrow (\tau_1,\tilde{v}_1,\tilde{\xi}_1),\dots , (\tau_n,\tilde{v}_n,\tilde{\xi}_n) 
\end{equation}
with 
\begin{equation*}
0= \tau_0 \leq \tau_1 \leq \ldots \tau_n \leq t
\end{equation*}
and where $\tilde{v}_i\in\R^3$ are the outgoing velocities of the tagged particle with respect to the forward flow after a collision suffered at time $\tau_i$ with an obstacle of velocity $\tilde{\xi}_i\in E(\tilde{v}_{i-1},\tilde{v}_i)$. 
Taking into account the Jacobian of the change of variables, i.e.
$$\frac{\ud{b}_{1}\ud{v}_{1}\ldots\ud{b}_{n}\ud{v}_{n}}{n!}=\varepsilon^{2n}\frac{\ud{\tau}_{1}\ldots\ud{\tau}_{n}\ud{\tilde{v}}_{1}\ldots\ud{\tilde{v}}_{n}\ud{\tilde{\xi}}_{1}\ldots\ud{\tilde{\xi}}_{n}}{\prod_{i=1}^{n}|\tilde{v}_{i}-\tilde{v}_{i-1}|},$$
(the factor $n!$ from picking the collision order for the obstacles) we can rewrite $h_\eps$, given by \eqref{eq:intermediate}, as follows:
\begin{align} \label{formula4}
 	\langle \phi,h_{\ep}(t)\rangle&= \int_{\R^3\times\R^3 } \ud{x} \ud{v}\, \sum_{n\geq 0}(1-\alpha)^{n}(\mu_\ep \eps^2)^n \int_{0}^{t}\uud{\tau}_1\ldots \int_{\tau_{n-1}}^{t}\uud{\tau}_n\int_{\R^3\times E(\tilde{v}_{0},\tilde{v}_1)} \uud{\tilde{v}_1}\uud{\tilde{\xi}_1}\ldots  		       \nonumber \\&
	\qquad \int_{\R^3\times E(\tilde{v}_{n-1},\tilde{v}_{n})} \uud{\tilde{v}_n}\uud{\tilde{\xi}_n}  \left( \prod_{i=1}^n \frac{\mathcal{M}_{\beta}(\tilde{\xi}_i)}{|\tilde{v}_i-\tilde{v}_{i-1}|}\right) \exp\left({-\mu\sum_{i=0}^n\lambda(\tilde{v}_i) (\tau_{i+1}-\tau_i)}\right) \nonumber 
	\\&	\qquad  \mathbf{1}_{\{(\mbf{b}_n,\mbf{u}_n)_{n}\in \bm{\Delta}_{n}(t)\}}  f_0(x,v) \phi(T^{t}_{\mbf{\eta_{\ep,n}}}(x,v)) .
\end{align}

Here the indicator function $\mathbf{1}_{\{(\mbf{b}_n,\mbf{u}_n)_{n}\in \bm{\Delta}_{n}(t)\}}  $ ensures that the map \eqref{change var} is one-to-one. Now we relabel the variables $t_i = \tilde{t}_{n-i} $, $v_i = \tilde{v}_{n-i} $, $\xi'_i = \tilde{\xi}_{n-i} $ with $\xi_{i}' \in E(v_i,v_{i-1})$. 
Notice that, with this relabeling, the right time simplex is: $0\leq t_{n}\leq t_{n-1}\leq \dots \leq t_2\leq t_1\leq t_0=t.$ 

This allows to construct for any $s\in [0,t]$ the backward limit flow $$\gamma^{-s}(x,v)=(x(-s),v(-s))$$ introduced in Definition \ref{def:limitflow}. For notational simplicity we skipped the dependence of the flow on the ordered obstacle configurations.

In the new variables $t_i,{v}_i,{\xi'}_i$  we can express 
$$\mathbf{1}_{\{(\mbf{b}_n,\mbf{u}_n)_{n}\in \bm{\Delta}_{n}(t)\}} =(1-\bm{\mathbf{1}}_{\rm{rec}})(1-\bm{\mathbf{1}}_{\rm{int}}),$$
where we defined the characteristic functions to exclude the set of pathological situations, namely the set of configurations delivering Recollisions and Interferences (cf. items $i), ii)$ in Definition \ref{def:pathologies}). 
More precisely, we recall that we defined $\bm{\mathbf{1}}_{\rm{rec}}$, $\bm{\mathbf{1}}_{\rm{int}}$ as
\begin{align}\label{def:chirec}
\bm{\mathbf{1}}_{\rm{rec}}=&\mathbf{1}_{\{\mbf{b}_N\;\text{s.t.}\;\text{i)}\;\text{is}\;\text{realized}\}}\\
\bm{\mathbf{1}}_{\rm{int}}=&\mathbf{1}_{\{\mbf{b}_N\;\text{s.t.}\;\text{ii)}\;\text{is}\;\text{realized}\}}\label{def:chiint}.
\end{align}

 Then, to simplify the notation, in what follows we denote by $\Omega'_n$ the full integration domain and  
$$ \int_{0}^{t_0}\uud{t}_1\int \uud{v}_1\int_{E(v_1,v)} \uud{\xi'_1}  \dots  \int_{0}^{t_{n-1}}\uud{t}_n \int \uud{v}_n \int_{E(v_n,v_{n-1})} \uud{\xi'_n} = \int_{\Omega'_n}\uud{t}_1 \uud{v}_1\uud{\xi'_{1}}\ldots \uud{t}_n \uud{v}_n\uud{\xi'_{n}} .$$
Then we can rewrite \eqref{formula4} as 
\begin{equation}
\begin{aligned} \label{formula5}
\langle \phi,h_{\ep}(t)\rangle&=\sum_{n=0}^\infty \int (1-\alpha)^{n}\mu^n \int_{\Omega'_n}\uud{t}_1 \uud{v}_1\uud{\xi'_{1}}\ldots \uud{t}_n \uud{v}_n\uud{\xi'_{n}}  \left( \prod_{i=1}^n \frac{\mathcal{M}_{\beta}({\xi}'_i)}{|v_i-v_{i-1}|} \right) \\
&\quad e^{- \mu \sum_{i=0}^{n-1} \lambda(v_i)(t_{i}-t_{i+1})}\, e^{- \mu\lambda(v_n)t_n}  \,(1-\bm{\mathbf{1}}_{\rm{rec}})(1-\bm{\mathbf{1}}_{\rm{int}}) f_0(\gamma^{-t}(x,v)) \phi(x,v) \ud{x} \ud{v} .
\end{aligned}
\end{equation}
This implies $\langle \phi,h_{\ep}(t)\rangle=\langle \phi,\bar{f}_{\ep}(t)\rangle$ for any $\phi \in C_b(\X) $ with $\phi(A)=0$. With the same change of variables we obtain 
 $h_{\ep}(t,\A)=\bar{f}_\eps(t,\A).$
 
 Using Proposition \ref{prop:deffeps} we obtain $h_{\ep}=\bar{f}_{\ep}$ as given in \eqref{formula6}-\eqref{formula6A}.
This concludes the proof of Proposition \ref{prop:step1strategy}.
\end{proof}

We will later prove that the characteristic functions \eqref{def:chirec}-\eqref{def:chiint} satisfy:
\begin{equation} \label{eq:badevents}
	(1-\bm{\mathbf{1}}_{\rm{rec}})(1-\bm{\mathbf{1}}_{\rm{int}})\rightarrow 1, \quad  \text{as $\eps \rightarrow 0$}.	
\end{equation}
Our goal is to make this quantitative.  To this end, we keep track of the mass that goes to the annihilated state $\A$. We recall that $\bar{f}_{\ep}(t,\A)$ is given by \eqref{formula6A} and  $f_{\ep}(t,\A)$ is given by \eqref{def:fepA}, i.e. $ f_{\ep}(t,\A)=1- \|f_{\ep}\|_{L^1(\Reals^3 \times \Reals^3)}$. We further define 	
\begin{align*}
	{f}(t,\A)&=\int_{\R^3\times \R^3} \ud{x} \ud{v}\, \sum_{n=0}^\infty  \left(1-(1-\alpha)^{n}\right)\mu^n \int_{\Omega'_n}\uud{t}_1 \uud{v}_1\uud{\xi'_{1}}\ldots \uud{t}_n \uud{v}_n\uud{\xi'_{n}}  \left( \prod_{i=1}^n \frac{\M_{\beta}(\xi'_i)}{|v_i-v_{i-1}|}\right)\nonumber \\& 
	\quad e^{- \mu \sum_{i=0}^{n-1} \lambda(v_i)(t_{i}-t_{i+1})}\,e^{- \mu\lambda(v_n)t_n}  \, f_0(\gamma^{-t}(x,v)), 
	\end{align*}
	in other words we have
\begin{equation*}
	f(t,\A)=1- \|f\|_{L^1(\Reals^3 \times \Reals^3)}
\end{equation*}
where $f$ is the solution of the annihiled linear Boltzman equation \eqref{eq:linearized} 
given by \eqref{eq:srsolB_3}.

We introduce the error function:
\begin{equation} \label{eq:psidef}
	\begin{aligned}
		{\psi}_\eps(t) &=\sum_{n\geq 0} \int_{\R^3\times\R^3}\ud{x}\ud{v} \, \mu^n  \int_{\Omega'_n}\uud{t}_1 \uud{v}_1\uud{\xi'_{1}}\ldots \uud{t}_n \uud{v}_n\uud{\xi'_{n}}   \left( \prod_{i=1}^n \frac{\M_{\beta}(\xi'_i)}{|v_i-v_{i-1}|}\right) \\
		& \quad e^{- \mu \sum_{i=0}^{n-1} \lambda(v_i)(t_{i}-t_{i+1})}  (\bm{\mathbf{1}}_{\rm{rec}}+\bm{\mathbf{1}}_{\rm{int}}) f_0(\gamma^{-t}(x,v))  .
	\end{aligned}
\end{equation}
Clearly, ${\psi}_\eps(t)=\|\bar{f}_\eps(t)-f(t)\|_{L^1(\X)} $. Notice that 
$$\|\bar{f}_\eps(t)-f(t)\|_{L^1(\X)}=\|\bar{f}_\eps(t)-f(t)\|_{L^1(\R^3\times \R^3)}+\left|\bar{f}_\eps(t,\A)-f(t,\A)\right|.$$
In the following we will prove an explicit bound for the error function $\psi_{\ep}$, that we state in the following proposition.
\begin{proposition}\label{lem:expliciterror}
Let ${\psi}_\eps(t)$ be defined as in \eqref{eq:psidef}. Let be $\kappa_0>0$ and $r\in (0,\frac 1 3)$. Then, for any $\kappa \in (0,\kappa_0)$ we have
\begin{align} \label{eq:L1rec}
{\psi}_\eps(t) \leq C \left[{\ep}^{\kappa} t (\mu +\mu^2)+ \eps^{r-2\kappa_0} [\eps^{\kappa_0}\mu(\mu t+(\mu t)^2)+((\mu t)^\frac52+(\mu t)^\frac72)]\right], \qquad \forall t \geq0.\end{align}
\end{proposition}

\begin{remark}
We remark that for any given $\mu>0$ the bound \eqref{eq:L1rec} becomes
\begin{align} \label{eq:L1rec_bis}
{\psi}_\eps(t) \leq C_{\mu} \left[{\ep}^{\kappa} t + \eps^{r-\kappa_0} ( t+t^2)+\eps^{r-2\kappa_0}( t^\frac 5 2+ t^\frac72)]\right].
\end{align}
Notice that we keep the $\mu$ dependence in the bound \eqref{eq:L1rec} as well as in the proof of Proposition~\ref{lem:expliciterror} since it is important when considering hydrodynamic rescalings of the system due to possible further rescalings of the intensity $\mu M_\eps \to \infty$ and of the timescale (cf. Section \ref{s:hydro}). 

\end{remark}

We observe that Proposition \ref{lem:expliciterror} guarantees that the non Markovian pathologies, namely the recollisions and interferences, are vanishing in the limit $\ep\to 0$. The proof is postponed to Section~\ref{s:rec}. 

\medskip

On the extended phase space $\X=\Reals^3 \times \Reals^3 \cup \{\A\}$ we have
$ f_\eps-\bar{f}_\eps\geq 0,\, f\geq 0$, and $f$, $f_\eps$ have unit total mass in $\X$, therefore:
\begin{align*}
1 = \|f_\eps(t)\|_{L^1(\X)} 	&= \| f_\eps - \bar{f}_\eps + f+ (\bar{f}_\eps-f)\|_{L^1(\X)} \\
&\geq \|f_\eps - \bar{f}_\eps\|_{L^1(\X)} + \|f \|_{L^1(\X)} - \|\bar{f}_\eps-f\|_{L^1(\X)} . 
\end{align*} 
We then have  
\begin{equation}\label{eq:esterr}
\|f_\eps - \bar{f}_\eps\|_{L^1(\X)}\leq \|\bar{f}_\eps-f\|_{L^1(\X)}  \leq  C_{\mu} \left[{\ep}^{\kappa} t + \eps^{r-\kappa_0} ( t+t^2)+\eps^{r-2\kappa_0}( t^\frac 5 2+ t^\frac72)]\right],
\end{equation}
where in the last inequality we used \eqref{eq:L1rec_bis}.
Therefore, using \eqref{eq:esterr} and applying  \eqref{eq:L1rec_bis} we get 
\begin{align}
	\|f_{\eps} -f \|_{L^1(\X)} &\leq  \|f_{\eps} - \bar{f}_\eps\|_{L^1(\X)}+ \|\bar{f}_\eps-f \|_{L^1(\X)}\nonumber\\&
	\leq C_{\mu} \left[{\ep}^{\kappa} t + \eps^{r-\kappa_0} ( t+t^2)+\eps^{r-2\kappa_0}( t^\frac 5 2+ t^\frac72)]\right].
\end{align}
This conludes the proof of Theorem \ref{th:MAIN1}.

\section{Explicit control of the non Markovian pathologies: proof of Proposition \ref{lem:expliciterror}} \label{s:rec}
Our goal is to estimate the occurence of trajectories with pathologies: recollisions and interferences. 
In the following we will assume that $f_0$ satisfies Assumption~\ref{ass:f0}, i.e. $f_0(x,v)=g_0(x,v) \mathcal{M}_{\beta}(v)$ with
\begin{align} \label{eq:gass}
	\int_{\Reals^3} \|g_0(x,\cdot)\|_{L^\infty(\Reals^3)} \ud{x} \leq C.
\end{align}

We first estimate the error function $\psi_\eps$ introduced in \eqref{eq:psidef}, using \eqref{eq:gass} for $f_0(\gamma^{-t}(x,v))$. This yields
\begin{align*}
	\psi_\eps(t) \leq  \sum_{n\geq 0} \mu^n \int_{\Reals^3} \ud{x}\int_{\Reals^3}& \ud{v} \,\mathcal{M}_{\beta}(v_n) \, \|g_0(\gamma^{-t}_1(x,v),\cdot)\|_{L^\infty(\Reals^3)} \int_{\Omega'_n}\uud{t}_1 \uud{v}_1\uud{\xi'_{1}}\ldots \uud{t}_n \uud{v}_n\uud{\xi'_{n}}   \\
	&\left( \prod_{i=1}^n \frac{\M_{\beta}(\xi'_i)}{|v_i-v_{i-1}|}\right) e^{- \mu \sum_{i=0}^{n-1} \lambda(v_i)(t_{i}-t_{i+1})}\,e^{- \mu\lambda(v_n)t_n}  (\bm{\mathbf{1}}_{\rm{rec}} +\bm{\mathbf{1}}_{\rm{int}}),
\end{align*}
where $\bm{\mathbf{1}}_{\rm{rec}}$ is the characteristic function of having a recollision event defined as in \eqref{def:chirec}. Now we use that $\bm{\mathbf{1}}_{\rm{rec}},\, \bm{\mathbf{1}}_{\rm{int}}$ are translation invariant in the variable $x$. We fix $x(0)=x_0$ with $x_0\in \R^3$ arbitrary and perform the integration in $x$:
\begin{align}\label{eq:bdPsi}
	\psi_\eps(t) \leq  C\sum_{n\geq 0} \mu^n \int_{\Reals^3}& \ud{v} \,\mathcal{M}_{\beta}(v_n) \,  \int_{\Omega'_n}\uud{t}_1 \uud{v}_1\uud{\xi'_{1}}\ldots \uud{t}_n \uud{v}_n\uud{\xi'_{n}}   \left( \prod_{i=1}^n \frac{\M_{\beta}(\xi'_i)}{|v_i-v_{i-1}|}\right) \nonumber\\
	&e^{- \mu \sum_{i=0}^{n-1} \lambda(v_i)(t_{i}-t_{i+1})}\,e^{- \mu\lambda(v_n)t_n}  (\bm{\mathbf{1}}_{\rm{rec}} +\bm{\mathbf{1}}_{\rm{int}}).
\end{align}
We now use the Maxwellian energy conservation identity 
for a collision $(v_i,\xi_i')\leftrightarrow (v_{i-1},\xi)$ to write 
$$\mathcal{M}_{\beta}(v_n) \prod_{i=1}^n \mathcal{M}_{\beta}(\xi'_i)= \mathcal{M}_{\beta}(v)\prod_{i=1}^n\left(\frac{\mathcal{M}_{\beta}(v_{i})}{\mathcal{M}_{\beta}(v_{i-1})} \mathcal{M}_{\beta}(\xi'_i)\right)=\mathcal{M}_{\beta}(v) \prod_{i=1}^n\left(\mathcal{M}_{\beta}(\xi_i)\right).$$
Then, we change variables $\xi'_i\to \xi_i$ denoting by $\Omega_n$ the new integration domain, i.e. we set 
 $$ \int_{0}^{t_0}\uud{t}_1\int \uud{v}_1\int_{E(v,v_1)} \uud{\xi}_1  \ldots  \int_{0}^{t_{n-1}}\uud{t}_n \int \uud{v}_n \int_{E(v_{n-1},v_{n})} \uud{\xi}_n = \int_{\Omega_n}\uud{t}_1 \uud{v}_1\uud{\xi}_1\ldots \uud{t}_n \uud{v}_n\uud{\xi}_n.$$
 This allows to rewrite \eqref{eq:bdPsi} as 
 \begin{align*} 
	\psi_\eps(t) \leq  C\sum_{n\geq 0} \mu^n \int_{\Reals^3}& \ud{v} \,\mathcal{M}_{\beta}(v) \,  \int_{\Omega_n}\uud{t}_1 \uud{v}_1\uud{\xi}_1\ldots \uud{t}_n \uud{v}_n\uud{\xi}_n  \left( \prod_{i=1}^n \frac{\M_{\beta}(\xi_i)}{|v_i-v_{i-1}|}\right) \nonumber\\
	&e^{- \mu \sum_{i=0}^{n-1} \lambda(v_i)(t_{i}-t_{i+1})}\,e^{- \mu\lambda(v_n)t_n}  (\bm{\mathbf{1}}_{\rm{rec}} +\bm{\mathbf{1}}_{\rm{int}}).
\end{align*}
The goal of this section is to prove Proposition \ref{lem:expliciterror} providing the quantitative estimate for the error function $\psi_{\ep}(t)$.  
In order to estimate  $\psi_\eps$, we distinguish between trajectories where the tagged particle is very fast and those where the velocity is controlled. 
\begin{definition} For any $\kappa_{0} >0$, set
$$\vartheta(w)=\begin{cases} 1  &\text{ if }\quad  \mathcal{M}_{\beta}(w) \geq \eps^{\kappa_{0}}\\
0 &\text{ otherwise }\end{cases}$$
and we define
\begin{align}
\bar{\psi}_\eps(t) =	&\int_{\Reals^3} \ud{v} \, \mathcal{M}_{\beta}(v) \sum_{n\geq 0}\mu^n  \int_{\Omega_n}\uud{t}_1 \uud{v}_1\uud{\xi}_1\ldots \uud{t}_n\uud{v}_n\uud{\xi}_n  \left( \prod_{i=1}^n \frac{\mathcal{M}_{\beta}(\xi_i)}{|v_i-v_{i-1}|}\vartheta(v_i)\right) \nonumber \\& 
\quad e^{- \mu \sum_{i=0}^{n-1} \lambda(v_i)(t_{i}-t_{i-1})}\,e^{- \mu\lambda(v_n)t_n} ( \bm{\mathbf{1}}_{\rm{rec}}+\bm{\mathbf{1}}_{\rm{int}} ) \label{eq:barpsi}  \\
\psi^\infty_\eps(t) = &\int_{\Reals^3} \ud{v}\, \mathcal{M}_{\beta}(v) \sum_{n\geq 0}\mu^n \int_{\Omega_n}\uud{t}_1 \uud{v}_1\uud{\xi}_1\ldots\uud{t}_n \uud{v}_n\uud{\xi}_n   \left( \prod_{i=1}^n \frac{\mathcal{M}_{\beta}(\xi_i)}{|v_i-v_{i-1}|}  \right) \left(1-\prod_{i=1}^n \vartheta(v_i) \right) \nonumber \\&  
\quad  e^{- \mu \sum_{i=0}^{n-1} \lambda(v_i)(t_{i}-t_{i-1})}\,e^{- \mu\lambda(v_n)t_n}  .  \label{eq:infpsi}
\end{align}\end{definition}
We have 
$$\psi_{\eps} \leq \psi^{\infty}_{\eps}+ \bar{\psi}_{\eps}$$ so it suffices to estimate separately $\psi^{\infty}_{\eps}$ and $\bar{\psi}_{\eps}$. We first prove that the contribution of  $\psi^\infty_\eps$ vanishes in the limit $\ep\to 0$, and it satisfies the following bound.

\begin{lemma}[Small probability for fast particles] \label{lem:fast}
	Let be $\kappa_0>0$. For  $\kappa \in (0,\kappa_0)$ there exists a constant $C>0$ 	such that we have
	\begin{align*}
	\psi^{\infty}_{\eps}(t) \leq C{\ep}^{\kappa} t (\mu +\mu^2).
	\end{align*}
\end{lemma}
The proof of Lemma \ref{lem:fast} will be presented in Subsection \ref{ss:fastpart}. \medskip

Next, we estimate the function $\bar{\psi}_\eps$. Our goal is to prove the following
\begin{lemma}[Small probability for non Markovian pathologies of slow particles] \label{lem:slow}
	Let be $\kappa_0>0$ and $r \in (0,\frac 1 3)$. Then there exists a constant $C>0$  such that we have
	\begin{align*}
	\bar{\psi}_\eps(t) \leq C \eps^{r-2\kappa_0} [\eps^{\kappa_0}\mu(\mu t+(\mu t)^2)+(\mu t)^\frac52+(\mu t)^\frac72 ].
	\end{align*}
\end{lemma}
The proof of Lemma \ref{lem:slow} is technical and will be presented in Subsection \ref{ss:lsrec}.
 
Before the proofs of the Lemmas above, we prove a preliminary result which allows to control the second moment of the number of collisions of the limit process. This result is the content of Subsection \ref{ss:2mom}.

\subsection{Moment estimate for the number of collision for the Boltzmann process}\label{ss:2mom}
In order to find good estimates for the function $\bar{\psi}_\eps$, we will 
make use of the fact that the Maxwellian $\mathcal{M}_{\beta}$ is a steady state of the spatially homogeneous Boltzmann process $V(t)$ without annihilation,  i.e. the Markov process with forward equation:
\begin{align*}
\partial_t f(t,v) = \mu Q(\mathcal{M}_{\beta},f(t,\cdot))(v).
\end{align*}
In this way we are able to make use of the loss term in the Boltzmann equation. 
To this end, let $\bm{N}_2(t,t')$ be given by  
\begin{align}\label{eq:n2}
  \bm{N}_2(t,t')= \E[\bm{N}^2(t,t')],\quad	\bm{n}_2(t)  = \bm{N}_2(t,0),
\end{align}
where $\bm{N}(t,t')$ denotes the number of collisions between times $0\leq t'\leq t$ of the Boltzmann process for the velocities $V(t)$ with initial distribution $V_0 \sim \M_{\beta}(v) \ud{v}$, and intensity $\mu$. 
Since the Maxwellian distributions are invariant under the evolution we have:
\begin{align} \label{eq:statmoment}
	\bm{N}_2(t,t')= \bm{n}_2(t-t').
\end{align}

Now our strategy is to first estimate $\bm{n}_2(t)$ for $t\leq \mu^{-1}$, and then extend the estimate using the stationarity of the process.

It will be useful for later purposes to introduce an auxiliary functional. For $n\in \N$ let be $\zeta=\zeta(t_0,v_0,t_1,v_1,\xi_1,\ldots,t_n,v_n,\xi_n)$ a function of the variables defining $n$ collisions. We now define:
\begin{equation}\begin{aligned}\label{eq:defpsinep}
	\Psi_{n,\eps}[\zeta](t) &=	\int_{\Reals^3} \ud{v}\, \mathcal{M}_{\beta}(v) \mu^n   \int_{\Omega_n}\uud{t}_1 \uud{v}_1\uud{\xi}_1\ldots \uud{t}_n \uud{v}_n\uud{\xi}_n  \left( \prod_{k=1}^n \frac{\M_{\beta}(\xi_k)}{|v_k-v_{k-1}|}\right) \\& 
\qquad e^{- \mu \sum_{k=0}^{n-1} \lambda(v_k)(t_{k}-t_{k+1})}\,e^{- \mu\lambda(v_n)t_n}\, \zeta. 
\end{aligned}\end{equation}

\begin{lemma}[Second moment of the number of collisions] \label{lem:momBoltz}
Let $\bm{n}_2(t)$ be given by \eqref{eq:n2}. Then the following identity holds for any $\eps>0$:
	\begin{align} \label{eq:n2psi}
		\bm{n}_2(t) = \sum_{n=0}^\infty n^2 \Psi_{n,\eps}[1].
	\end{align}
Moreover, there exists a constant $C>0$  such that
	for $t\in [0,\infty]$ we can estimate:
	\begin{align} \label{eq:longmoment}
	\bm{n}_2(t) \leq C(\mu t+(\mu t)^2). 	
	\end{align}
\end{lemma}

\begin{proof}	We write $\bm{n}_2(t)$ by using the explicit representation of the expectation in \eqref{eq:n2}: 
	\begin{align*}
	\bm{n}_2(t) &=  \sum_{n\geq 0} \int_{\Reals^3} \ud{v}\, \mathcal{M}_{\beta}(v) 	\mu^n  n^2   \int_{\Omega_n} \uud{t}_1 \uud{v}_1\uud{\xi}_1\ldots \uud{t}_n \uud{v}_n\uud{\xi}_n\\ 
	 &\quad \left(\prod_{k=1}^n \frac{\M_{\beta}(\xi_k)}{|v_k-v_{k-1}|} \right)e^{- \sum_{k=0}^{n-1} \mu  \lambda(v_{k})(t_{k}-t_{k+1})}\,e^{- \mu\lambda(v_n)t_n} .	
	\end{align*}
	We observe that this explicit expansion of $\bm{n}_2$ proves  \eqref{eq:n2psi}. We now prove \eqref{eq:longmoment}.
	Since $\lambda(v)>0$, we can estimate $e^{- \sum_{k=0}^{n-1} \mu  \lambda(v_{k})(t_{k}-t_{k+1})- \mu\lambda(v_n)t_n}\leq 1$
	and perform the integral on the time simplex to obtain:
	\begin{align*}
	\bm{n}_2(t) 	&\leq  \int_{\Reals^3}\ud{v} \, \mathcal{M}_{\beta}(v) 	\sum_{n\geq 0} \,\frac{(\mu  t)^n}{n!}  n^2   \int_{\Reals^3 }\ud{v_1} \dots     \int_{\Reals^3 } \ud{v_n}\left(\prod_{k=1}^n \frac{e(v_{k-1},v_{k})}{|v_k-v_{k-1}|} \right) ,
	\end{align*}
	where we used \eqref{eq:IntEst} to handle the integrals with respect to the variable $\xi_k$. 
	Now we estimate iteratively the integrals in $v_k$:
	\begin{align}\label{eq:bdn2}
	\bm{n}_2(t)	&\leq  \int_{\Reals^3} \mathcal{M}_{\beta}(v) 	\left(\sum_{n\geq 0}\frac{(\mu t)^n  n^2}{n!} C^n (1+|v|)^n \right)\ud{v} \nonumber\\
	&\leq  \int_{\Reals^3} \mathcal{M}_{\beta}(v) C(1+|v|) \mu  t + \left(C(1+|v|)\mu  t\right)^2\sum_{n\geq 2}\frac{(\mu t)^{n-2}  }{(n-2)!} C^n (1+|v|)^{n-2} \ud{v}\nonumber	 \\
	&\leq C\int_{\Reals^3} \mathcal{M}_{\beta}(v)  \left((1+|v|) \mu  t + C\left( (1+|v|)\mu  t\right)^2e^{C^*\mu  (1+|v|)t}\right) \ud{v}, 
	\end{align}
	where in the first inequality we used that 
	$$\int_{\R^3}\ud{v'}\,  (1+|v|)^k \frac{e(v,v')}{|v-v'|} =\int_{|v'|\leq 2|v|}\ud{v'}\,(1+|v|)^k  \frac{e(v,v')}{|v-v'|}+\int_{|v'|\geq 2|v|}\ud{v'} \,  (1+|v|)^k \frac{e(v,v')}{|v-v'|} $$
	with
	$$ \int_{|v'|\leq 2|v|}\ud{v'}\,(1+|v|)^k  \frac{e(v,v')}{|v-v'|} \leq C (1+|v|)^{k+1}$$	
	and 
	$$ \int_{|v'|\geq 2|v|}\ud{v'} \,  (1+|v|)^k \frac{e(v,v')}{|v-v'|} \leq \tilde {C}^{k} \left(\frac{k+1}{2}\right)! \quad \text{for}\; k\, \text{odd}, $$
	$$ \int_{|v'|\geq 2|v|}\ud{v'} \,  (1+|v|)^k \frac{e(v,v')}{|v-v'|} \leq \tilde {C}^{k} \left(\frac{k+2}{2}\right)!\quad \text{for}\; k\, \text{even}.$$
	Then, \eqref{eq:bdn2} proves the estimate \eqref{eq:longmoment} for $t \in [0,(C^*\mu) ^{-1}]$.  	
 We now extend the estimate to arbitrary times. For any $k \in \N$ and for any $t_{\ast}>0$ we can estimate:
	\begin{align*}
		\bm{n}_2(k t_*) 	&= \E\left[ \bm{N}^2(k t_*,0)\right] \\&
		=\E\left[\left(\sum_{\ell=1}^k n(\ell t_*,(\ell -1)t_*)\right)^2\right]
		= \E\left[\left(\sum_{\ell=1}^k \sum_{j=1}^k \bm{N}(\ell t_*,(\ell -1)t_*)\bm{N}(j t_*,(j -1)t_*)\right)\right]  \\& \leq  \frac 1 2 \E\left[\left(\sum_{\ell=1}^k \sum_{j=1}^k \big( \bm{N}^2(\ell t_*,(\ell -1)t_*)+\bm{N}^2(j t_*,(j -1)t_*)\big)\right)\right]
		\leq k^2 \bm{n}_2(t_*). 
	\end{align*}
	where we used in the first inequality Young's inequality and in the second one \eqref{eq:statmoment}.
	Picking $t_* = (C^*\mu)^{-1}$ we obtain the claim.
\end{proof}

\subsection{Proof of Lemma \ref{lem:fast}}\label{ss:fastpart}

In order to prove Lemma \ref{lem:fast} we first define $R_{\ep}=R(\ep,\kappa_0)=\ep^{-\kappa_0}$ and $\ell_\ep=\ell (\ep,\kappa_0)$ such that $\M_{\beta}(\ell_{\ep})=\ep^{\kappa_0}$. We observe that $\psi^{\infty}_\eps$ defined as in \eqref{eq:infpsi} can be estimated by:
\begin{equation}\label{eq:splitinfpsi}
\psi^{\infty}_\eps\leq \psi^{\infty}_{\eps,1}+\psi^{\infty}_{\eps,2}+\psi^{\infty}_{\eps,3}
\end{equation}
where 
\begin{align}\label{eq:infpsi1}
 \psi^{\infty}_{\eps,1}(t) = 
&\int_{\Reals^3} \ud{v}\, \mathcal{M}_{\beta}(v) \sum_{n\geq 0}\mu^n \int_{\Omega_n}\uud{t}_1 \uud{v}_1\uud{\xi}_1\ldots \uud{t}_n \uud{v}_n\uud{\xi}_n\left( \prod_{i=1}^n \frac{\M_{\beta}(\xi_i)}{|v_i-v_{i-1}|}  \right) \mathbf{1}_{\{n\geq 2\}} \nonumber \\&
\quad  e^{- \mu \sum_{i=0}^{n-1} \lambda(v_i)(t_{i}-t_{i+1})}\,e^{- \mu\lambda(v_n)t_n} ,
\end{align}
and
\begin{align}\label{eq:infpsi2}
 \psi^{\infty}_{\eps,2}(t) = &\int_{\Reals^3} \ud{v}\, \mathcal{M}_{\beta}(v) \sum_{n\geq 0}\mu^n \int_{\Omega_n}\uud{t}_1 \uud{v}_1\uud{\xi}_1\ldots \uud{t}_n \uud{v}_n\uud{\xi}_n   \left( \prod_{i=1}^n \frac{\M_{\beta}(\xi_i)}{|v_i-v_{i-1}|}  \right) \mathbf{1}_{\{|v|\geq \frac{ \ell_{\ep}}{2}\}} \nonumber \\&
\quad \mathbf{1}_{\{n\leq 1\}} e^{- \mu \sum_{i=0}^{n-1} \lambda(v_i)(t_{i}-t_{i+1})}\,e^{- \mu\lambda(v_n)t_n}, 
\end{align}
and 
\begin{align}\label{eq:infpsi3}
 \psi^{\infty}_{\eps,3}(t) = &\int_{\Reals^3} \ud{v}\, \mathcal{M}_{\beta}(v) \sum_{n\geq 0}\mu^n \int_{\Omega_n}\uud{t}_1 \uud{v}_1\uud{\xi}_1\ldots \uud{t}_n \uud{v}_n\uud{\xi}_n   \left( \prod_{i=1}^n \frac{\M_{\beta}(\xi_i)}{|v_i-v_{i-1}|}  \right) \mathbf{1}_{\{|v|\leq \frac{ \ell_{\ep}}{2},\,|v_1 |\geq   \ell_{\ep}\}} \nonumber \\&
\quad \mathbf{1}_{\{n\leq 1\}} e^{- \mu \sum_{i=0}^{n-1} \lambda(v_i)(t_{i}-t_{i+1})}\,e^{- \mu\lambda(v_n)t_n} .
\end{align}
We consider \eqref{eq:infpsi1}. Using the same argument of the proof of Lemma \ref{lem:momBoltz} we get
\begin{align}\label{eq:infpsi1_bis}
 \psi^{\infty}_{\eps,1}(t) & = \int_{\Reals^3} \ud{v}\, \mathcal{M}_{\beta}(v) \sum_{n\geq 2}\mu^n \int_{\Omega_n}\uud{t}_1 \uud{v}_1\uud{\xi}_1\ldots \uud{t}_n \uud{v}_n\uud{\xi}_n   \left( \prod_{i=1}^n \frac{\M_{\beta}(\xi_i)}{|v_i-v_{i-1}|}  \right) \nonumber \\&
\qquad  e^{- \mu \sum_{i=0}^{n-1} \lambda(v_i)(t_{i}-t_{i+1})}\,e^{- \mu\lambda(v_n)t_n}  \nonumber \\&
\leq \int_{\Reals^3} \ud{v}\, \mathcal{M}_{\beta}(v) \left(\sum_{n\geq 2}\frac{(\mu t)^n  }{n!} C^n (1+|v|)^n \right) \leq \int_{\Reals^3} \ud{v}\, \mathcal{M}_{\beta}(v) (\mu t C(1+|v|))^2 e^{\mu t C(1+|v|)}.
\end{align}
We now choose $t=R_{\ep}^{-1}$. Then, from \eqref{eq:infpsi1_bis}, for $\ep>0$ sufficiently small we obtain 
\begin{align} \label{eq:psi1final}
 \psi^{\infty}_{\eps,1}(R_{\ep}^{-1})\leq (R_{\ep}^{-1} \mu C)^2. 
 \end{align}
We now  proceed by estimating  the function $\psi^{\infty}_{\eps,2}$ defined in \eqref{eq:infpsi2}:
\begin{align*}
 \psi^{\infty}_{\eps,2}(t) 	&= \int_{\Reals^3} \ud{v}\, \mathcal{M}_{\beta}(v)\mu \int_{\Omega_1}\uud{t}_1\uud{v}_1\uud{\xi_1} 
 \frac{\M_{\beta}(\xi_1)}{|v_1-v|} \mathbf{1}_{\{|v|\geq \ell_{\ep}/2\}} e^{- \mu (\lambda(v_0)(t-t_{1})-\lambda(v_1)t_1)} \\
 							&\leq C \mu t  \int_{\Reals^3} \ud{v}\, \mathcal{M}_{\beta}(v) (1+|v|) \mathbf{1}_{\{|v|\geq \ell_{\ep}/2\}}.
\end{align*}
Then for any $\kappa \in (0,\kappa_0)$ we obtain
\begin{align} \label{eq:psi2final}
	\psi^{\infty}_{\eps,2}(t) \leq C \mu t  \eps^{\kappa}.  	
\end{align}
It remains to estimate $\psi^{\infty}_{\eps,3}$ (cf. \eqref{eq:infpsi3}). To this end we use
that for $|v|\leq \ell_\ep/2$ and $v_1\geq \ell_{\ep}$ we have
\begin{align}
	e(v,v_1) =  \mathcal{M}_{\beta}\left(v_1 \cdot \frac{v_1-v}{|v_1-v|}\right) \leq \mathcal{M}_{\beta} \left(\frac12 v_1\right). 
\end{align}
Therefore we can bound the contribution of $\psi^{\infty}_{\eps,3}$ by:
\begin{align*}
 \psi^{\infty}_{\eps,3}(t) = &\int_{\Reals^3} \ud{v}\, \mathcal{M}_{\beta}(v)\mu \int_{\Omega_1}\uud{t}_1\uud{v}_1\uud{\xi}_1
 \frac{\M_{\beta}(\xi_1)}{|v_1-v|} \mathbf{1}_{\{|v|\leq \frac{ \ell_{\ep}}{2},\,|v_1 |\geq   \ell_{\ep}\}}e^{- \mu (\lambda(v_0)(t-t_{1})-\lambda(v_1)t_1)}   \\
 \leq & \, t\int_{\Reals^3} \ud{v}\, \mathcal{M}_{\beta}(v)\mu \int_{\Reals^3} \uud{v}_1
 \frac{e(v,v_1)}{|v_1-v|} \mathbf{1}_{\{|v|\leq \frac{ \ell_{\ep}}{2},\,|v_1 |\geq   \ell_{\ep}\}} \\
 \leq &\, t\int_{\Reals^3} \ud{v}\, \mathcal{M}_{\beta}(v)\mu \int_{\Reals^3} \uud{v}_1
 \frac{\mathcal{M}_{\beta}(\frac12 v_1)}{|v_1-v|} \mathbf{1}_{\{|v|\geq \ell_{\ep}\}}.
\end{align*}
As before, for $\kappa \in(0,\kappa_0)$ we obtain:
\begin{align}\label{eq:psi3final}
	\psi^{\infty}_{\eps,3}(t) \leq  C \mu t  \eps^\kappa. 	
\end{align}
Inserting the estimates \eqref{eq:psi1final}-\eqref{eq:psi3final} into \eqref{eq:splitinfpsi} yields for $\kappa \in (0,\kappa_0)$:
\begin{align}\label{eq:psiinftyshort}
	\psi^{\infty}_{\eps}(R_{\ep}^{-1}) \leq C (\mu+\mu^2) R_{\ep}^{-1} \eps^\kappa.  
\end{align}
Exploiting the stationarity property (cf. Section \ref{ss:2mom}) of the process we obtain the claim of Lemma~\ref{lem:fast}:
\begin{align} \label{eq:psiinftyfinal}
\psi^{\infty}_{\eps}(t)\leq \frac{t}{R_{\ep}^{-1}}  \psi^{\infty}_{\eps}(R_{\ep}^{-1})  \leq C{\ep}^{\kappa} t (\mu +\mu^2). 
\end{align}

\subsection{Proof of Lemma \ref{lem:slow}}\label{ss:lsrec} 

The first step of the proof is similar to the one in \cite{BNPP}. We estimate the total probability of each pathological event (recollision or interference) by estimating the probability of a pathology for each possible pair $(i,j)$ of obstacles, i.e. 
for $n\in \N$ and $1\leq i<j\leq n$ we write:
$$\left(\prod_{k=0}^n \vartheta(v_k) \right) \bm{\mathbf{1}}_{\rm{rec}} \leq \sum_{i} \sum_{j>i} \bm{\mathbf{1}}_{\rm{rec}}^{i,j},\qquad \left(\prod_{k=0}^n \vartheta(v_k) \right)\bm{\mathbf{1}}_{\rm{int}}\leq \sum_{i} \sum_{j>i} \bm{\mathbf{1}}_{\rm{int}}^{i,j},$$
where
\begin{align}
	\bm{\mathbf{1}}_{\rm{rec}}^{i,j}&=\bm{\mathbf{1}}_{\rm{rec}}^{i,j}(t_0,v_0,t_1,v_1, \ldots, t_n,v_n; \xi_i) \nonumber \\&
	 = \prod_{k=0}^n \vartheta(v_k) \mathbf{1}_{\{\left(x(-s)\in \partial B_\eps(b_i(-s),\ep), \text{ for some  $s\in(t-t_{j},t-t_{j+1})$}\right)\}}, \label{eq:chirecij} \\
	\bm{\mathbf{1}}_{\rm{int}}^{i,j}&=\bm{\mathbf{1}}_{\rm{int}}^{i,j}(t_0,v_0,t_1,v_1,\ldots, t_n,v_n;\xi_j) \nonumber \\&
	= \prod_{k=0}^n \vartheta(v_k) \mathbf{1}_{\{\left( x(-s) \in B(b_j(-s),\ep) \text{ for some } s\in(t-t_{i},t-t_{i+1})\right)\}}.\label{eq:chiintij} 
\end{align}
We recall that $\vartheta(w)= \mathbf{1}_{\{\mathcal{M}_{\beta}(w)\geq \eps^{\kappa_0}\}},\quad \text{with}\;\kappa_0>0$.

It is useful to further distinguish the cases when the distance of the tagged
particle to the $i$th obstacle is small and when it is large at the last collision event before the recollision. To this end we introduce:
\begin{align}\label{eq:bdpatij}
	\bm{\mathbf{1}}_{\rm{rec}}^{i,j}	&\leq  \bm{\mathbf{1}}_{\rm{rec},-}^{i,j} +\bm{\mathbf{1}}_{\rm{rec},+}^{i,j}, \quad
	\bm{\mathbf{1}}_{\rm{rec},\pm }^{i,j}	= \bm{\mathbf{1}}_{\rm{rec}}^{i,j}\cdot \mathbf{1}_{\{|x(-(t-t_j))-b_i(-(t-t_j))|\cdot ({\pm 1})\geq \ \eps^\gamma \cdot ({\pm 1})\}}, \\ 
	\bm{\mathbf{1}}_{\rm{int}}^{i,j}	&\leq  \bm{\mathbf{1}}_{\rm{int},-}^{i,j} +\bm{\mathbf{1}}_{\rm{int},+}^{i,j}, \quad
	\bm{\mathbf{1}}_{\rm{int},\pm }^{i,j}	= \bm{\mathbf{1}}_{\rm{int}}^{i,j}\cdot \mathbf{1}_{\{|x(-(t-t_j))-b_i(-(t-t_j))|\cdot ({\pm 1})\geq \  \eps^\gamma \cdot ({\pm 1})\}}, 
	\end{align}
where $\gamma\in (0,1)$ will be fixed later. 

Then, using the definition of $\Psi_{n,\eps}$ given as in \eqref{eq:defpsinep}, we can estimate $\bar{\psi}_\eps$ given as in \eqref{eq:barpsi} by 
\begin{align}\label{eq:splpsibar}
	\bar{\psi}_\eps(t) 	&\leq \sum_{n=0}^\infty \sum_{1\leq i<j \leq n}\Psi_{n,\eps}[\bm{\mathbf{1}}^{i,j}_{\rm{rec},+} +\bm{\mathbf{1}}^{i,j}_{\rm{int},+}](t) + \sum_{n=0}^\infty \sum_{1\leq i<j \leq n}\Psi_{n,\eps}[\bm{\mathbf{1}}^{i,j}_{\rm{rec},-}+\bm{\mathbf{1}}^{i,j}_{\rm{int},-}](t)   .
\end{align}
\begin{lemma} [Long distance recollisions/interferences] \label{lem:long}
Let be $\kappa_0>0$ and $\gamma\in (0,1)$ be as in \eqref{eq:bdpatij}. 
	For $\gamma' \in (0,1-\gamma)$ there exists a constant $C>0$ such that for every $n\in \N$ and $1\leq i<j\leq n$ there holds:
	\begin{align} \label{longbound}
	   \Psi_{n,\eps}[\bm{\mathbf{1}}^{i,j}_{\rm{rec},+}](t) &\leq C  \eps^{\gamma'-2\kappa_0} (\mu  t)^{\frac32} \Psi_{n,\eps}[1](t), \\
	   \Psi_{n,\eps}[\bm{\mathbf{1}}^{i,j}_{\rm{int},+}](t) &\leq C  \eps^{\gamma'-2\kappa_0} (\mu  t)^{\frac32} \Psi_{n,\eps}[1](t). \label{longbound2}
	\end{align}
\end{lemma}
\begin{proof}
	We prove the estimate \eqref{longbound}, since the proof of \eqref{longbound2} uses the same geometrical argument.  
	Consider a long-distance recollision between the $i$-th and $j$-th collision event, with $j>i$. 	At the collision time $t_{j}$, the $i$-th obstacle and the tagged particle
	have a distance of:
	\begin{align*}
	|x_{\rm{rel}}(-(t-t_{j}))|= |(x-b_i)(-(t-t_j))|\geq \eps^\gamma,
	\end{align*}
	since we consider the case of a long-distance recollision.
	Therefore  their relative velocity  can only vary over a cone around $x_{\rm{rel}}$:
	\begin{align*}
	v_j-u_i &\in C_{\eps^{1-\gamma}} (x_{\rm{rel}}(-(t-t_j)), \text{ where:} \\
	C_{r} (x) &= \left\{w \in \Reals^3: \left|\arccos\left(\frac{x\cdot w}{|x||w|}\right)\right|\leq r\right\}.
	\end{align*}
	For every $\gamma'\in (0,1-\gamma)$ there exists $\eps_0=\eps_0(\gamma')>0$ such that for any $0<\eps <\eps_0(\gamma')$, we have 
	\begin{align}
	C_{\eps^{1-\gamma}} (x) \cap \{w\in \Reals^3: \mathcal{M}_{\beta}(w)\geq \eps^{\kappa_0}\} \subset \Cyl_{\eps^{\gamma'}}(x), 
	\end{align}
	where 
	\begin{align}
	\Cyl_r(x) := \{w\in \Reals^3: \operatorname{dist}(w,x\Reals)\leq r\}
	\end{align}
	is the cylinder with radius $r$ and axis given by the line generated by $x$. 
	For $n\in \N$, $j>i$, consider: $\Psi_{n,\eps}[\bm{\mathbf{1}}^{i,j}_{\rm{rec},+}](t)$:
	\begin{align*}
	\Psi_{n,\eps}[\bm{\mathbf{1}}^{i,j}_{\rm{rec},+}](t) &=	\int_{\Reals^3} \ud{v}\, \mathcal{M}_{\beta}(v) \mu ^n  \int_{\Omega_n}\uud{t}_1 \uud{v}_1\uud{\xi}_1\ldots \uud{t}_n\uud{v}_n\uud{\xi}_n \left( \prod_{k=1}^n \frac{\M_{\beta}(\xi_k)}{|v_k-v_{k-1}|}\right) \\ &
	\quad e^{- \mu \sum_{k=0}^{n-1}\lambda(v_{k})(t_k-t_{k+1}) } \,e^{- \mu \lambda(v_{n})t_n } \bm{\mathbf{1}}^{i,j}_{\rm{rec},+}.
	\end{align*}
	We now inspect the integral in $v_j$, for all other integration variables arbitrary but fixed. Since $\bm{\mathbf{1}}_{\rm{rec}}^{i,j}$ is independent on $\xi_k$ for any $k\neq i$ (cf.~\eqref{eq:chirecij}), we can perform, using \eqref{eq:IntEst}, all the integrals in $\xi_k$ but $\xi_i$. Then, we get
	\begin{align*}
	\Psi_{n,\eps}[\bm{\mathbf{1}}^{i,j}_{\rm{rec},+}](t) &=	\int_{\Reals^3} \ud{v}\, \mathcal{M}_{\beta}(v) \mu ^n  \int_{0}^{t}\uud{t}_1 \int_{\R^3}\uud{v}_1\ldots  \int_{0}^{t_{i-1}}\uud{t}_{i} \int_{\R^3}\uud{v}_i\int_{E(v_{i-1},v_{i})}\uud{\xi}_{i}\frac{\M_{\beta}(\xi_i)}{|v_i-v_{i-1}|}\ldots \\&
	\quad \int_{0}^{t_{n-1}}\uud{t}_n \int_{\R^3}\uud{v}_n  \left( \prod_{k\neq i } \frac{e(v_{k-1},v_{k})}{|v_k-v_{k-1}|}\right) e^{- \mu \sum_{k=0}^{n-1}\lambda(v_{k})(t_k-t_{k+1}) } \,e^{- \mu \lambda(v_{n})t_n } \bm{\mathbf{1}}^{i,j}_{\rm{rec},+}.
	\end{align*}
	For simplicity we write $\Delta v_j= |v_j-v_{j-1}|$, $\Delta t_j=t_j-t_{j+1}$. Using the geometrical considerations above we then have:
	\begin{align*}
	&\int_{\Reals^3} \frac{e(v_{j-1},v_j)e(v_{j},v_{j+1})}{|\Delta v_j| |\Delta v_{j+1}|} e^{-\mu \lambda(v_j) \Delta t_j} \bm{\mathbf{1}}^{i,j}_{\rm{rec}} \vartheta(v_{j-1})\vartheta(v_j) \vartheta(v_{j+1})\ud{v_j}\\
	\leq &C  e^{-\mu \lambda(0) \Delta t_j}\int_{\Cyl_{\eps^{\gamma'}}(x_{\rm{rel}}(-(t-t_{j})))+\xi_i} \frac{1}{|\Delta v_j| |\Delta v_{j+1}|}   \vartheta(v_{j-1}) \vartheta(v_j) \vartheta(v_{j+1})  \ud{v_j},
	\end{align*}
	where we used the monotonicity of the function $\lambda$ and the boundedness of $e(\cdot,\cdot)$. 
	Using the localization of $v_j$, for every $\gamma'\in (0,1-\gamma)$ we can then bound the integral,using that $\vartheta(v_{j-1})\vartheta(v_{j+1})\leq 1$,  by:
	\begin{align}\label{eq:longvj}
	\int_{\Reals^3} \frac{e(v_{j-1},v_j)e(v_{j},v_{j+1})}{|\Delta v_j| |\Delta v_{j+1}|} e^{-\mu \lambda(v_j) \Delta t_j} \bm{\mathbf{1}}^{i,j}_{\rm{rec},+} \vartheta(v_{j-1})\vartheta(v_j) \vartheta(v_{j+1})\ud{v_j}		
	\leq C \eps^{\gamma'} e^{-\mu \lambda(0) \Delta t_j}.
	\end{align}
	We now compare the estimate \eqref{eq:longvj} to the value of the integral
	without $\bm{\mathbf{1}}^{i,j}_{\rm{rec},+}$:
	\begin{align*}
	&\int_{\Reals^3} \frac{e(v_{j-1},v_j)e(v_{j},v_{j+1})}{|\Delta v_j| |\Delta v_{j+1}|} e^{-\mu \lambda(v_j) \Delta t_j} \vartheta(v_{j-1}) \vartheta(v_j) \vartheta(v_{j+1})\ud{v_j}   	
	\\&
	\geq 	C \eps^{2\kappa_0}\int_{\Reals^3} \frac{1}{|\Delta v_j| |\Delta v_{j+1}|} e^{-\mu \lambda(v_j) \Delta t_j} \vartheta(v_{j-1}) \vartheta(v_j) \vartheta(v_{j+1}). 
	\end{align*}
	For $|v_j|\leq (\mu t)^{-\frac12}$ we can estimate $e^{-\mu \lambda(v_j) \Delta t_j}\geq c e^{-\mu \lambda(0) \Delta t_j} $. Moreover, thanks to the $ \vartheta(v_{j-1}),  \vartheta(v_{j+1})$  we have that $|v_{j-1}|,\, |v_{j+1}|\leq C\,|\log(\ep)|$. Therefore, for $\kappa \in (0,2\kappa_0)$ we obtain the lower bound:
	\begin{align*}
		&\int_{\Reals^3} \frac{e(v_j,v_{j-1})e(v_{j},v_{j+1})}{|\Delta v_j| |\Delta v_{j+1}|} e^{-\mu \lambda(v_j) \Delta t_j}  \vartheta(v_{j-1})\vartheta(v_j) \vartheta(v_{j+1})\ud{v_j} 	\geq  C \eps^{\kappa} (\mu t)^{-\frac32} e^{-\mu \lambda(0) \Delta t_j}.
	\end{align*}
	Combining the preceding estimates we obtain:
	\begin{align*}
		\Psi_{n,\eps}[\bm{\mathbf{1}}^{i,j}_{\rm{rec},+}](t) \leq C \eps^{\gamma'-2\kappa_0} (\mu  t)^{\frac32} {\Psi}_{n,\eps}[\vartheta(v_{j-1}) \vartheta(v_j) \vartheta(v_{j+1})](t)\leq C \eps^{\gamma'-2\kappa_0} (\mu  t)^{\frac32} {\Psi}_{n,\eps}[1](t).	
	\end{align*}
\end{proof}

\begin{lemma}[Short distance recollisions/interferences] \label{lem:short} 
Let be $\kappa_0>0$ and $\gamma\in (0,1)$ be as in \eqref{eq:bdpatij}. 
	For $\gamma'\in (0,\frac12 \gamma )$ there exists a constant $C>0$ such that for every $n\in \N$ and $1\leq i<j\leq n$ there holds:
		\begin{align} \label{shortbound}
		{\Psi}_{n,\eps}[\bm{\mathbf{1}}^{i,j}_{\rm{rec},-}] (t)\leq  C \eps^{\gamma'-\kappa_0}\left( \eps^{-\kappa_0}(\mu t)^\frac32{\Psi}_{n,\eps}[1](t) + \mu{\Psi}_{n-1,\eps}[1](t)  \right), \\
		{\Psi}_{n,\eps}[\bm{\mathbf{1}}^{i,j}_{\rm{int},-}] (t)\leq  C \eps^{\gamma'-\kappa_0}\left( \eps^{-\kappa_0}(\mu t)^\frac32{\Psi}_{n,\eps}[1](t) + \mu{\Psi}_{n-1,\eps}[1](t)  \right). \label{shortbound2}
	\end{align}
\end{lemma}
\begin{proof}
	The proofs of \eqref{shortbound} and \eqref{shortbound2} are almost identical, we show the argument for the case of recollisions.
	Let $n\in \N$, $1\leq i< j\leq n$. We consider 
	\begin{align}\label{eq:chirec-}
	\Psi_{n,\eps}[\bm{\mathbf{1}}_{\rm{rec},-}^{i,j}](t) &=	\int_{\Reals^3} \ud{v}\, \mathcal{M}_{\beta}(v) \mu ^n  \int_{\Omega_n}\uud{t}_1 \uud{v}_1\uud{\xi}_1\ldots \uud{t}_n\uud{v}_n\uud{\xi}_n \left( \prod_{k=1}^n \frac{\M_{\beta}(\xi_k)}{|v_k-v_{k-1}|}\right)\nonumber \\ &
	\quad e^{- \mu \sum_{k=0}^{n-1}\lambda(v_{k})(t_k-t_{k+1}) } \,e^{- \mu \lambda(v_{n})t_n } \bm{\mathbf{1}}_{\rm{rec},-}^{i,j}\nonumber \\&
	=\int_{\Reals^3} \ud{v}\, \mathcal{M}_{\beta}(v) \mu ^n  \int_{0}^{t}\uud{t}_1 \int_{\R^3}\uud{v}_1\dots  \int_{0}^{t_{i-1}}\uud{t}_{i} \int_{\R^3}\uud{v}_i\int_{E(v_{i-1},v_{i})}d\xi_{i}\frac{\M_{\beta}(\xi_i)}{|v_i-v_{i-1}|}\dots\nonumber \\&
	\quad \int_{0}^{t_{n-1}}\uud{t}_n \int_{\R^3}\uud{v}_n  \left( \prod_{k\neq i } \frac{e(v_{k-1},v_{k})}{|v_k-v_{k-1}|}\right) e^{- \mu \sum_{k=0}^{n-1}\lambda(v_{k})(t_k-t_{k+1}) } \,e^{- \mu \lambda(v_{n})t_n } \bm{\mathbf{1}}_{\rm{rec},-}^{i,j}.
	\end{align}
	where using that $\bm{\mathbf{1}}_{\rm{rec}}^{i,j}$ is independent on $\xi_k$ for any $k\neq i$ (cf.~\eqref{eq:chirecij}), and using \eqref{eq:IntEst} 	we performed all the integrals in $\xi_k$ but $\xi_i$.

	 For simplicity write $\Delta v_j= |v_j-v_{j-1}|$, $\Delta t_j=t_j-t_{j+1}$. Furthermore, we fix $n\in \N$ and all $v_k,u_k$, $k\neq j$ and derive an estimate for the $j$-th collision integral that is uniform in these parameters. 
	 We now distinguish two cases : when the relative velocity $|v_j-u_{i,j}|$ is small and when is big. Then we can rewrite \eqref{eq:chirec-} as 
	 \begin{equation}\label{eq:splitchirec-}
	 \Psi_{n,\eps}[\bm{\mathbf{1}}_{\rm{rec},-}^{i,j}](t)=\Psi_{n,\eps}[\bm{\mathbf{1}}^{i,j}_{\rm{rec},-}\mathbf{1}_{\{|v_j-u_{i,j}|\leq \eps^{\frac12 \gamma}\}}](t)+\Psi_{n,\eps}[\bm{\mathbf{1}}^{i,j}_{\rm{rec},-}\mathbf{1}_{\{|v_j-u_{i,j}|\geq \eps^{\frac12 \gamma}\}}](t).
	 \end{equation}
	We first consider recollisions with relative velocity $|v_j-u_{i,j}|\leq \eps^{\frac12 \gamma}$ small. 
	Here we denote by $u_{i,j}$ the velocity of the obstacle $i$ evaluated at the colision time $t_j$ (backward). Then the integral in $v_j$ in $\Psi_{n,\eps}[\bm{\mathbf{1}}^{i,j}_{\rm{rec},-}\mathbf{1}_{\{|v_j-u_{i,j}|\leq \eps^{\frac12 \gamma}\}}](t)$ defined in \eqref{eq:splitchirec-} can be estimated by:
	\begin{align*}
		&\int_{\Reals^3}  \frac{e(v_{j-1},v_{j})e(v_{j},v_{j+1})}{\Delta v_j \Delta v_{j+1}} e^{- \mu \lambda(v_j)\Delta t_j} 	 \bm{\mathbf{1}}^{i,j}_{\rm{rec},-} \mathbf{1}_{\{|v_j-   u_{i,j}|\leq \eps^{\frac12 \gamma}\}}\ud{v_j} \\&
		\leq C e^{- \mu \lambda(0)\Delta t_j} \int_{B_{\eps^{\frac{\gamma}{2} }}(u_{i,j})}  \frac{1}{\Delta v_j \Delta v_{j+1}} \ud{v_j} 
		\leq C	e^{- \mu \lambda(0)\Delta t_j} \eps^{\frac12 \gamma},
	\end{align*}
	where we used the monotonicity of $\lambda$, the boundedness of $e(\cdot,\cdot)$ and the restriction on $v_j$ given by the characteristic function.
	
	Comparing the estimate  above to the value of the integral without $\bm{\mathbf{1}}^{i,j}_{\rm{rec},-} $, arguing as in the proof of Lemma \ref{lem:long}, we have that for $\gamma' \in (0,\frac12 \gamma)$ :
	\begin{align*}
	&\int_{\Reals^3}  \frac{e(v_{j-1},v_{j})e(v_{j},v_{j+1})}{\Delta v_j \Delta v_{j+1}} e^{- \mu \lambda(v_j)\Delta t_j} 	 \bm{\mathbf{1}}^{i,j}_{\rm{rec},-} \mathbf{1}_{\{|v_j-u_i|\leq \eps^{\frac12 \gamma}\}}\ud{v_j}\\ &\leq C	 \eps^{\gamma'-2\kappa_0} (\mu t)^\frac32 \int d v_j \frac{e(v_{j-1},v_{j})}{\Delta v_{j} \Delta v_{j+1}} e^{- \mu \lambda(v_j)\Delta t_j} 	.
	\end{align*}
	Therefore we have:
	\begin{align}\label{eq:splitchi1}
		\Psi_{n,\eps}[ \bm{\mathbf{1}}^{i,j}_{\rm{rec},-}\mathbf{1}_{\{|v_j-u_{i,j}|\leq \eps^\gamma\}}](t) \leq C \eps^{\gamma'-2\kappa_0}(\mu t)^\frac32 \Psi_{n,\eps}[1],
	\end{align}
	where $\kappa_0>0.$
	
	Now consider the case of large relative velocity, i.e. 
	$|v_j-u_{i,j}|\geq \eps^{\frac12 \gamma}$. To this end, since on the support of $\bm{\mathbf{1}}^{i,j}_{\rm{rec},-}$ it holds that $\displaystyle |t_{j}-t_{j+1}|\leq \frac{\eps^{\gamma}}{|v_j-u_{i,j}|}$, we can  use:
	\begin{align}\label{eq:estchisd}
	\bm{\mathbf{1}}^{i,j}_{\rm{rec},-}\mathbf{1}_{\{|v_i-u_{i,j}|\geq \eps^{\frac12 \gamma}\}} \leq \mathbf{1}_{\{|t_j-t_{j+1}|\leq \eps^{\frac12 \gamma}\}}.
	\end{align}
	Now we freeze all integrals in $\Psi_{n,\eps}[\bm{\mathbf{1}}^{i,j}_{\rm{rec},-}\mathbf{1}_{\{|v_j-u_{i,j}|\geq \eps^{\frac12 \gamma}\}}](t)$ defined in \eqref{eq:splitchirec-} except the ones in $t_j$ and $v_j$: 
	\begin{align}\label{eq:estinteg}
	&\int_{t_{j+1}}^{t_{j-1}}\uud{t}_{j}\int \uud{v_j} \left(\prod_{k=j-1,j}\frac{e(v_{k},v_{k+1})}{\Delta v_{k+1}} e^{- \mu \lambda(v_k)\Delta t_k} \vartheta(v_{k})\right) \bm{\mathbf{1}}^{i,j}_{\rm{rec},-}	\mathbf{1}_{\{|v_j-u_{i,j}|\geq \eps^\frac12\}} \nonumber \\
	&\leq  	\frac{1}{|v_{j-1}-v_{j+1}|}\int_{t_{j+1}}^{t_{j+1}+\eps^{\frac12 \gamma}}\uud{t}_{j}\int \left(\frac{1}{\Delta v_j}+\frac{1}{\Delta v_{j-1}}\right) \prod_{k=j-1,j} e(v_{k+1},v_{k}) e^{- \mu \lambda(v_k)\Delta t_k}  \vartheta(v_k),
	\end{align}
	where we used \eqref{eq:estchisd} to restrict the time integral and we used  triangular inequality to control the relative velocities.
	For $|v|\leq \log|\eps|$ we have $\lambda(v)\leq C |\log \eps|$ due to the linear
	growth of $\lambda$. 
	Then, we obtain:
	\begin{align} \label{eq:timeest}
	 & \prod_{k=j-1,j}\vartheta(v_k)e^{- \mu \lambda(v_{k})(t_k-t_{k+1})} \mathbf{1}_{\{|t_{j}-t_{j+1}|\leq \eps^{\frac12 \gamma}\}}\nonumber \\&
	  \leq \vartheta(v_{j-1})\vartheta(v_{j})e^{- \mu \lambda(v_{j-1})(t_{j-1}-t_{j})} \mathbf{1}_{\{|t_{j}-t_{j+1}|\leq \eps^{\frac12 \gamma}\}}\nonumber\\&
	  \leq \vartheta(v_{j-1})\vartheta(v_{j})e^{- \mu \lambda(v_{j-1})(t_{j-1}-t_{j+1})} e^{ \mu \lambda(v_{j-1})(t_{j}-t_{j+1})}\mathbf{1}_{\{|t_{j}-t_{j+1}|  \leq \eps^{\frac12 \gamma}\}} \nonumber\\& 
	  \leq \vartheta(v_{j})e^{- \mu \lambda(v_{j-1})(t_{j-1}-t_{j+1})} e^{\mu C |\log \eps| \eps^{\frac12 \gamma}} 
	 \leq C e^{-\mu \lambda(v_{j-1})(t_{j-1}-t_{j+1})}\vartheta(v_{j}).
	\end{align}
	We further observe that on the support of $\vartheta(v_{j-1})$ we have $e(v_{j-1},v_{j+1})^{-1}\leq C \eps^{-\kappa_0}$. Combining this with \eqref{eq:timeest}, from \eqref{eq:estinteg} we obtain that for $\gamma'\in (0,\frac12 \gamma)$:
	\begin{align*}
	&\int_{t_{j+1}}^{t_{j-1}}\uud{t}_{j}\int \uud{v}_j \left(\prod_{k=j-1,j}\frac{e(v_{k},v_{k+1})}{|v_{k+1}-v_k|} e^{- \mu \lambda(v_k)\Delta t_k} \vartheta(v_{k-1})\right) \bm{\mathbf{1}}^{i,j}_{\rm{rec},-}	\mathbf{1}_{\{|v_j-u_i|\geq \eps^\frac12\}}\\
	\leq &\, C e^{-\lambda(v_{j-1})(t_{j-1}-t_{j+1})} \frac {1}{|v_{j-1}-v_{j+1}|} \int_{t_{j+1}}^{t_{j+1}+\eps^{\frac12 \gamma}} \uud{t}_{j}\int \uud{v}_j \vartheta(v_{j})\left(\frac{1}{\Delta v_j}+\frac{1}{\Delta v_{j-1}}\right) \\
		\leq &\, C	\eps^{\gamma'}\frac{e^{-\lambda(v_{j-1})(t_{j-1}-t_{j+1})}}{|v_{j-1}-v_{j+1}|}\leq \, C	\eps^{\gamma'-\kappa_0}\frac{e^{-\lambda(v_{j-1})(t_{j-1}-t_{j+1})}e(v_{j+1},v_{j-1})}{|v_{j-1}-v_{j+1}|}.
	\end{align*}
We now consider ${\Psi}_{n,\eps}[\bm{\mathbf{1}}^{i,j}_{\rm{rec},-}\mathbf{1}_{\{|v_j-u_i|\geq \eps^\gamma\}}](t)$ given as in \eqref{eq:splitchirec-} and apply the estimate above. 
Then, relabeling the integration variables $(t_{j+1},v_{j+1})\to (t_j,v_j)$, $(t_{j+2},v_{j+2})\to (t_{j+1},v_{j+1})$, for $\gamma'\in (0,\frac12 \gamma)$ we get:
\begin{align*}
{\Psi}_{n,\eps}[\bm{\mathbf{1}}^{i,j}_{\rm{rec},-}\mathbf{1}_{\{|v_j-u_i|\leq \eps^\gamma\}}](t) 
 \leq &\, C \eps^{\gamma'-\kappa_0}\mu^n \int_{\Reals^3} \ud{v}\, \mathcal{M}_{\beta}(v)   \int_{\Omega_{n-1}}\uud{t}_1\uud{v}_1\uud{\xi}_1\ldots\uud{t}_{n-1} \uud{v}_{n-1}\uud{\xi}_{n-1}\\ &
  \left( \prod_{k=1}^{n-1} \frac{\M_{\beta}(\xi_k)}{|v_k-v_{k-1}|}\right) e^{- \mu \sum_{k=0}^{n-2}\lambda(v_{k})(t_k-t_{k+1}) } \,e^{- \mu \lambda(v_{n-1})t_{n-1} }
	\\ =& \, C \eps^{\gamma'-\kappa_0}  \mu {\Psi}_{n-1,\eps}[1](t).
\end{align*} 
	Hence, we have shown
	\begin{align}\label{eq:estchirec-short}
		{\Psi}_{n,\eps}[\bm{\mathbf{1}}^{i,j}_{\rm{rec},-}] &\leq {\Psi}_{n,\eps}[\bm{\mathbf{1}}^{i,j}_{\rm{rec},-}\mathbf{1}_{\{|v_j-u_i|\leq \eps^\gamma\}}] + {\Psi}_{n,\eps}[\bm{\mathbf{1}}^{i,j}_{\rm{rec},-}\mathbf{1}_{\{|v_j-u_i|\geq \eps^\gamma\}}] \nonumber \\
			&\leq C \eps^{\gamma'-\kappa_0}\left(\eps^{-\kappa_0}(\mu t)^\frac32{\Psi}_{n,\eps}[1](t) + \mu{\Psi}_{n-1,\eps}[1](t)  \right),
	\end{align}
	as claimed. We observe that it is straightforward to prove that for the estimate for the interference events we get the same estimate as in \eqref{eq:estchirec-short} which gives \eqref{shortbound2}.
\end{proof}
\medskip

\subsection{Proof of Proposition~\ref{lem:expliciterror}} 
We conclude this section with the proof of the estimate of pathologic trajectories.

	We choose $\gamma=2/3$. We recall that $\bar{\psi}_\eps$ given as in \eqref{eq:barpsi} satisfies \eqref{eq:splpsibar}. Then, using the Lemmas \ref{lem:long}, \ref{lem:short}  we have for any $r\in \left(0,\frac 1 3\right)$
		\begin{align*}
		\bar{\psi}_\eps(t) 	&\leq \sum_{n=0}^\infty \sum_{1\leq i<j \leq n}{\Psi}_{n,\eps}[\bm{\mathbf{1}}^{i,j}_{\rm{rec},+} +\bm{\mathbf{1}}^{i,j}_{\rm{int},+}](t) + \sum_{n=0}^\infty \sum_{1\leq i<j \leq n}{\Psi}_{n,\eps}[\bm{\mathbf{1}}^{i,j}_{\rm{rec},-}+\bm{\mathbf{1}}^{i,j}_{\rm{int},-}](t)   \\
							&\leq C  \eps^{r-2\kappa_0} \sum_{n=0}^\infty n^2  \left( \mu \ep^{\kappa_0}{\Psi}_{n-1,\eps}[1](t)+(\mu t)^\frac32 {\Psi}_{n,\eps}[1](t)\right)  .
	\end{align*}
	Now we observe that $\sum_{n=0}^\infty n^2 {\Psi}_{n,\eps}[1](t) = \bm{n}_2(t)$, so Lemma~\ref{lem:momBoltz} gives: 
	\begin{align} \label{eq:psibar}	
		\bar{\psi}_\eps(t) \leq C \eps^{r-2\kappa_0} [\eps^{\kappa_0}\mu(\mu t+(\mu t)^2)+(\mu t)^\frac52+(\mu t)^\frac72 ].	
		\end{align}
	The function $\psi(t)$ we can estimate by:
	\begin{align*}
		\psi(t) \leq \psi^\infty(t) + \bar{\psi}(t),
	\end{align*}
	so combining \eqref{eq:psibar} and Lemma~\ref{lem:fast} the claim of the Theorem follows.

\bigskip

\section{Outlook on the long-time behaviour of the Boltzmann-Rayleigh model} \label{s:hydro}

As we discussed in the Introduction, for annihilation parameter $\alpha=0$ our model reduces to the classical ideal Rayleigh Gas. Then, the kinetic description is given by the Boltzmann equation 
\begin{equation*}\label{eq:noalinearized}
\left\{\begin{array}{ll}\vspace{2mm}
	\partial_t f + v \cdot\nabla_x f = \mu Q(\mathcal{M}_{\beta},f), &\\
	f(x,v,0)=f_0(x,v).&
\end{array}\right.
\end{equation*}	
Moreover, for this model it is possible to look at a longer time scale, in which the diffusive behaviour of the classical ideal Rayleigh gas is described by the heat equation for the mass density. 
 
We can recover the same behaviour, if we consider the following rescaling (in the same spirit of the one proposed for the Lorentz model in \cite{BNP, BNPP}): 
\begin{align} \label{hydroscale}
 \mu_\eps = \mu \eps^{-2} M_\eps. 
\end{align}
Here $M_\eps \rightarrow \infty$ for $\eps\rightarrow 0$, more precisely we assume that $M_\eps$ is such that $\mu_\eps \eps^3\to 0$ and $\mu_\eps \eps^2\to \infty$ 
and satisfies $\lim_{\ep\to 0}{\psi}_\eps(M_{\ep}\tau) =0$ with ${\psi}_\eps$ the error function defined as in \eqref{eq:psidef}. Notice that Proposition \ref{lem:expliciterror} implies that $M_{\ep}$ can diverge algebraically in $\ep$. In the scaling limit \eqref{hydroscale} 
the gas is slightly more dense than in the standard Boltzmann-Grad scaling \eqref{scaling:B}. 

Through the scaling limit \eqref{hydroscale} we obtain at the kinetic level, a description given by the following Boltzmann equation: 

\begin{equation}\label{eq:hydroboltz}
\left\{\begin{array}{ll}\vspace{2mm}
	\partial_t f + v \cdot \nabla_x f = M_\eps \mu Q(\mathcal{M}_{\beta},f), &\\
	f(x,v,0)=f_0(x,v).&
\end{array}\right.
\end{equation}	
This allows to obtain the hydrodynamical description on a longer time scale, i.e. through a further rescaling of time $t\to M_\eps \tau$, when the annihilation parameter $\alpha$ vanishes. 

\medskip

\begin{proposition}[Hydrodynamics without annihilation] \label{th:MAIN2}
	Let be $f_{\eps}$ defined as in \eqref{def:fep}. We define $\tilde{f}_\eps$ by:
	\begin{align}\label{eq:fepidro}
	\tilde{f}_\eps (\tau,x,v) = {f}_\eps(M_\eps \tau,x,v).
	\end{align} 
	In the scaling limit \eqref{hydroscale} we have for any $R>0$
	\begin{equation*}
	\lim_{\ep\to 0}\tilde{f}_\eps(\tau,\cdot,\cdot)=\varrho(\tau,\cdot) \mathcal{M}_{\beta}(v) \quad \text{in}\; L^{\infty}([0,T];L^{1}(B_R(0) \times \R^3))
	\end{equation*}
	where $\varrho(t,x)$ satisfies the following heat equation:
	\begin{equation}\label{heateq}
	\left\{\begin{array}{ll}\vspace{2mm}
	\partial_t \varrho = D \Delta \varrho,&\\
	\varrho(x,0)=\varrho_0(x).&
	\end{array}\right.
	\end{equation}	
	Here $D$ is the diffusion coefficient given by the Green-Kubo formula:
	$$D=C \int_{\R^3}\ud{v}\, \M_{\beta}(v) v L^{-1}v,$$
	with $C>0$ a numerical constant and ${L}^{-1}$ the pseudo-inverse on the subspace $(\mathrm{Ker} L )^{\bot}$ of the linear Boltzmann operator 
	\begin{equation}\label{eq:LBoltzop}
	   L g (v)= \int_{\mathbb R^3}\!  dv_1\,\mathcal{M}_{\beta}(v_1)\int_{\mathbb{S}^{2}}\!d\hat n \,\big[\hat n\cdot(v-v_1)\big]_+\big[g(v')-g(v)  \big].
	\end{equation}
	\end{proposition}
\medskip

The proof of Proposition \ref{th:MAIN2} is standard and relies on the Hilbert expansion technique. We refer for instance to \cite{BGS, BNP, BNPP, CIP, EP}.
\medskip

\medskip

It is interesting to investigate the diffusive limit of the ideal Rayleigh gas with annihilation, i.e. when the annihilation parameter $\alpha>0$ does not vanish, since new effects and a different hydrodynamic equation can be expected. This will be discussed in a forthcoming paper. The following scenario is the one we expect. 

At the level of the particle system, we consider the rescaling 
\begin{align} \label{hydroscaleanni}
M_\eps^s \alpha_\eps = \alpha \;\text{with}\; s>0,\quad
 \mu_\eps = \mu \eps^{-2} M_\eps, 
\end{align}
with $M_\eps \rightarrow \infty$ for $\eps\rightarrow 0$ as required above. 
The gas is slightly more dense than in the standard Boltzmann-Grad scaling \eqref{scaling:B} as in the case of the scaling \eqref{hydroscale}, but here we further rescale the annihilation parameter $\alpha$. Depending on the exponent $s$ we can expect different asymptotic behaviours for $\tilde{f}_\eps$. We further emphasize, even if we restrict ourselves in the present paper to the case of hard-sphere interactions, that it is interesting to understand how the choice of particle interaction affects the limiting equation. Recall that, for classical diffusive behaviour described in Theorem \ref{th:MAIN2}, the choice of the interaction potential affects only the value of the diffusion coefficient $D$ via the Green-Kubo formula. Recall that we denote by $\LL$ the linear Boltzmann operator for hard-sphere interaction and now we introduce $\LL_{0}$ as the linear Boltzmann operator for Maxwellian interaction, i.e.
$$\LL_{0}f(v)=\lambda_{0}\int_{\R^{3}}\ud{v}_{1}\int_{\mathbb{S}^{2}}\frac{[(v-v_{1})\cdot \hat{n}]_{+}}{|v-v_{1}|}\left(\M_{\beta}(v')f(v')-\M_{\beta}(v_{1})f(v)\right)\ud{\hat{n}}.$$
For this interaction, the collision frequency is \emph{constant}, i.e.~$\lambda(v)=\lambda_0$.

At the hydrodynamic level, we expect the following scenario:
\begin{itemize}
\item if $s=2$ in the scaling limit \eqref{hydroscaleanni} and with Maxwellian interactions, which lead to a kinetic description given by the collision operator $\LL_{0}$, we expect that $\tilde{f}_\eps$ given as in \eqref{eq:fepidro} satisfies
	\begin{equation*}
	\lim_{\ep\to 0}\tilde{f}_\eps(\tau,\cdot,\cdot)=\varrho(\tau,\cdot) \mathcal{M}_{\beta}(v),
	\end{equation*}
where the function $\varrho(\tau,x)$ solves the equation
	\begin{align} \label{heateqmassloss}
	\partial_t \varrho+ r \varrho  = D \Delta \varrho, \quad D>0, r>0\ .
	\end{align}
\item if $s=2$ in the scaling limit \eqref{hydroscaleanni} and in the case of hard-sphere interactions, as the ones considered in this paper, the hydrodynamic behaviour is more difficult to describe since we still expect to recover a diffusion equation similar to \eqref{heateqmassloss} but now the dissipative term on the left hand side of \eqref{heateqmassloss} could have a more complicated form. Indeed, we observe that in \eqref{heateqmassloss} the constant term is expected since $\LL_{0}$ has a constant collision frequency. Hence, the case of hard-sphere interaction deserves a more detailed analysis. 

\item if $s>2$ in the scaling limit \eqref{hydroscaleanni} we expect that
	\begin{equation*}
	\lim_{\ep\to 0}\tilde{f}_\eps(\tau,\cdot,\cdot)=\varrho(\tau,\cdot) \mathcal{M}_{\beta}(v)
	\end{equation*}
where the function $\varrho(t,x)$ satisfies \eqref{heateq}.
\item if $0\leq s<2$  the functions  $M_{\eps}^{2-s}\tilde{f}_\eps(\tau,\cdot,\cdot)$ are of order one, but convergence and possible asymptotics for $\ep\to0$ are open problems
that deserve further investigation.
\end{itemize}

\bigskip

\textbf{Acknowledgements. }   
A.N. and R.W. acknowledge support through the
CRC 1060 \textit{The mathematics of emergent effects} of the University of
Bonn that is funded through the German Science Foundation (DFG).

\end{document}